\documentclass[journal]{IEEEtran}
\usepackage[utf8]{inputenc}
\usepackage{bm}
\usepackage{amsmath,epsfig,amssymb,dsfont,amsthm}
\usepackage{hyperref}
\usepackage{footmisc}
\usepackage{verbatim}
\usepackage{caption}
\usepackage{subcaption}
\usepackage[dvipsnames]{xcolor}
\usepackage{graphicx}
	\graphicspath{{./figures/}}
\usepackage{float}
\usepackage{url}
\usepackage{eurosym}
\usepackage{setspace}
\usepackage{balance}
\usepackage{dsfont}
\usepackage{bbm}
\usepackage{array}
\usepackage{cite}
\usepackage[ruled,vlined]{algorithm2e}
\usepackage{thmtools}
\usepackage{comment}
\usepackage[normalem]{ulem}
\usepackage{mathtools}
\usepackage{framed}

\usepackage{tabularx,booktabs}
\newcolumntype{A}{<{\raggedright\arraybackslash}X}
\newcolumntype{B}{>{\raggedleft\arraybackslash}{\hsize=.3\hsize}X}
\newcolumntype{C}{>{\hsize=.24\hsize}X}

\newcommand{\bA}{\boldsymbol{A}}
\newcommand{\bEdge}{\boldsymbol{E}}

\newcommand{\C}{\mathcal C}

\newcommand{\bT}{\mathsf{T}}
\newcommand{\bE}{\mathbb{E}}

\newcommand{\bmu}{\boldsymbol \mu}
\newcommand{\bpsi}{\boldsymbol \psi}
\newcommand{\bzeta}{\boldsymbol \zeta}
\newcommand{\btheta}{\boldsymbol \theta}

\newcommand{\Var}{\textrm{Var}}

\theoremstyle{plain}
\newtheorem{Thm}{Theorem}

\newtheorem{Lem}{Lemma}
\newtheorem{Asm}{Assumption}

\theoremstyle{remark}
\newtheorem{Rem}{Remark}
\newtheorem{Exa}{Example}

\newcommand{\qedsymb}{\hfill\ensuremath{\blacksquare}}                 
\allowdisplaybreaks
\title{Social Learning in Community Structured Graphs} 
\author{\IEEEauthorblockN{Valentina Shumovskaia, Mert Kayaalp, and Ali H. Sayed}\\
$\newline$
\IEEEauthorblockA{
École Polytechnique Fédérale de Lausanne (EPFL)
}
\thanks{Emails: $\{$valentina.shumovskaia, mert.kayaalp, ali.sayed$\}$@epfl.ch. The authors are with the Institute of Electrical and Micro Engineering at EPFL, Lausanne, Switzerland. A short version of this manuscript appears in the conference publication~\cite{shumovskaia2024distributed}.
}
}

\IEEEoverridecommandlockouts
\begin{document}

\maketitle

\begin{abstract}
    Traditional social learning frameworks consider environments with a homogeneous state, where each agent receives observations conditioned on that true state of nature. In this work, we relax this assumption and study the distributed hypothesis testing problem in a heterogeneous environment, where each agent can receive observations conditioned on their own personalized state of nature (or truth). We particularly focus on community structured networks, where each community admits their own true hypothesis. This scenario is common in various contexts, such as when sensors are spatially distributed, or when individuals in a social network have differing views or opinions. We show that the adaptive social learning strategy is a preferred choice for nonstationary environments, and allows each cluster to discover its own truth.
\end{abstract}

\begin{IEEEkeywords}
Social learning, hypotheses testing, diffusion strategy, adaptive learning, multitask learning, personalized learning.
\end{IEEEkeywords}

\section{Introduction and Related Work}

    The social learning framework~\cite{jadbabaie2012non, zhao2012learning, salami2017social, nedic2017fast, molavi2017foundations, molavi2018theory, bordignon2020adaptive, bordignon2022partial, lalitha2018social, acemoglu2011bayesian, gale2003bayesian, inan2022social, 9132712, 9670665} is a popular non-Bayesian approach for solving distributed hypothesis testing problems. In these problems, the objective is to learn and track an unknown underlying state of nature (or hypothesis) from streaming data. In particular, social learning can be used to model the opinion formation process over social networks, where  agents exchange beliefs on a chosen topic to arrive at a collaborative conclusion. For example, the work~\cite{shumovskaia2023discovering} considers a special application of social learning to identify influential users over Twitter. In traditional social learning algorithms, at each iteration, agents update their confidence level about each possible hypothesis. They do so by updating a probability mass function called the \textit{belief} vector. There are many variations of the social learning procedure, but the general framework consists of two repetitive steps. First, agents receive private observations conditioned on the true state of the environment and perform a local Bayesian update. Subsequently, the agents average the intermediate beliefs of their neighbors using either linear~\cite{zhao2012learning, jadbabaie2012non} or geometric~\cite{nedic2017fast, lalitha2018social, bordignon2020adaptive} combination rules. The key objective in all these and other related works is to allow the individual agents to learn the truth eventually, in which case the largest confidence component of the final belief vector for each individual agent would correspond to the true hypothesis. Under reasonable technical conditions, these works have established that truth learning is indeed possible. 

    The non-Bayesian social learning approach is an effective alternative to finding the fully Bayesian estimate of the unknown hypotheses. This is because the Bayesian solution is generally NP-hard and requires full knowledge of the graph topology and the likelihood models~\cite{gale2003bayesian, acemoglu2011bayesian, hkazla2021bayesian}. The non-Bayesian approach considered herein, and in the aforementioned works, relies on a fully decentralized implementation, allowing each agent to keep their observations private. There are many useful variations of non-Bayesian social learning strategy in the literature. For example, many works consider discrete hypotheses sets~\cite{jadbabaie2012non, nedic2017fast, molavi2017foundations, molavi2018theory, bordignon2020adaptive, bordignon2022partial, lalitha2018social}, while other works~\cite{9670665, 9670668} allow for compact hypotheses sets, which are useful for inference or regression tasks involving real-valued parameters. Other variants~\cite{inan2022social, cirillo2023role, bordignon2022partial, cemri} study different constraints on belief sharing, such as partial information sharing~\cite{bordignon2022partial}, learning under randomized collaborations~\cite{inan2022social}, or asynchronous learning~\cite{cemri}. Some other works such as~\cite{shumovskaia2022explainability, shumovskaia2023discovering, kayaalp2023causal} focus on explainability and on inverse questions, such as identifying the most influential agent in a network based on the sequence of beliefs. 
    In addition, the adaptive social learning strategy (ASL)~\cite{bordignon2020adaptive} infuses these methodologies with important adaptation and tracking abilities and allows agents to track drifts in the underlying hypothesis. This method resolves the stubbornness issue that plagues traditional solutions where agents resist changing their opinions.
 
    All these previous works assume that the agents receive observations arising from the same state of nature (or same hypothesis). In this work, we relax this assumption and allow agents to receive observations from different models (or hypotheses). In this context, we focus on studying the limiting behavior of the decision-making process and examine how the presence of agents with different state affects collaborative learning.

    For graphs involving a multiplicity of models, one useful topology is that of a community-structured graph such as the Stochastic Block Model (SBM)~\cite{abbe2017community, rohe2011spectral, lee2019review, goldenberg2010survey, fortunato2010community}. This model has been used for the analysis of exchanges over social networks such as Twitter~\cite{goldenberg2010survey, fortunato2010community,6750167, cherepnalkoski2016retweet}, and has also been found to be suitable to study opinion polarization over social networks~\cite{chitra2020analyzing, chitra2019understanding, bhalla2023local}. 
    In addition, the model can be used to study measurements from  spatially distributed sensors where the combination matrix describes geographic ``closeness" of the sensors~\cite{michaels2008detection, matias2018semiparametric}. We will show that, for such block models, traditional strategies as in~\cite{nedic2017fast,lalitha2018social,jadbabaie2012non,zhao2012learning} converge to a common state for all agents that best describes the observations (and this state is not necessarily optimal for the individual agents). In contrast, in the ASL~\cite{bordignon2020adaptive} strategy, the step-size parameter $\delta$, initially introduced to track distribution shifts, turns out to serve an important additional role. Under certain conditions, it will be shown to allow each cluster in the network to discover its own true hypothesis (or model).

    The community-structured model is reasonably close to the multi-task setting studied for the solution of optimizaton problems over graphs in~\cite{chen2015adaptive, chen2015diffusion, plata2017heterogeneous, zhao2015distributed, nassif2020multitask}. In the work~\cite{nassif2020multitask}, each agent optimizes their own objective function, and nodes belonging to the same cluster are known to attain the same optimal solution. The use of regularization promotes relationships among different tasks (or models), thus allowing the agents to benefit not only from cooperation from agents within the same cluster, but also from agents belonging to different clusters. The key difference in relation to our work is that we now focus on social learning, as opposed to decentralized optimization. In social learning, the intermediate step involves a Bayesian update propagating probability vectors, while in decentralized optimization that step is typically a stochastic gradient step involving (unnormalized) real-valued vectors. 
    
    One other work on social learning that also considers different hypotheses at the agents is~\cite{10124228}, however, the focus there is different from our formulation. The main contribution there is to discover if neighboring agents are observing data from different models and to cut their links when this happens. The intent is to limit agents to interact only with neighbors sharing similar models. The objective in the current work is different and also more general. We do not alter the topology, or cut links, but rather aim to discover all agent models despite the heterogeneity in the data and despite the interactions.

    We summarize the contributions of this work as follows:
    \begin{enumerate}
        \item In Section~\ref{sec:truth_learning}, we study the limiting behavior of opinion formation under the assumption of a heterogeneous environment. We show that the traditional social learning strategies~\cite{nedic2017fast,lalitha2018social,jadbabaie2012non,zhao2012learning} force all agents to converge to the same model, while the adaptive social learning (ASL) strategy~\cite{bordignon2020adaptive} allows for a more diverse behavior. 
        \item Section~\ref{sec:sbm} is devoted to SBM graphs where we derive conditions on $\delta$ that allow each community to converge to their own truth.
    \end{enumerate}

\section{Social Learning Model}\label{sec:model}
    Consider a collection of agents denoted by $\mathcal N$, interacting to form belief vectors that reflect their confidence in each possible hypothesis $\theta \in \Theta$ based on their private streaming observations and on interactions with their immediate neighbors. These interactions are governed by a combination graph $A \in [0,1]^{|\mathcal N| \times |\mathcal N|}$, where nonzero elements indicate an edge between two nodes. For any two connected agents, the value $a_{\ell k} = [A]_{\ell, k} > 0$ determines the level of trust that agent $k \in \mathcal N$ assigns to information arriving from agent $\ell \in \mathcal N$. The combination matrix is assumed to be left-stochastic, i.e., the weights on every column add up to one:
    \begin{align}
        \sum_{\ell \in \mathcal N} a_{\ell k} = 1,\; \forall k \in \mathcal N.
    \end{align}
    In this way, the total trust that each agent assigns to its neighbors is equal to $1$. Additionally, the combination graph is assumed to be strongly connected, which means that there exists at least one self-loop with a positive weight, i.e., $a_{kk}>0$ for some $k$ and, moreover,  there exists a path with positive weights between any two nodes. This condition implies that the matrix $A$ is primitive, i.e., there exists an integer power $t > 0$ such that all entries of $[A^t]_{k,\ell}$ are positive~\cite{sayed_2023}. These assumptions are common in the literature and they have been shown to guarantee the global truth learning for homogeneous environments under traditional social learning \cite{nedic2017fast, lalitha2018social, zhao2012learning}. It follows from the primitiveness of $A$ and the Perron-Frobenius theorem~\cite[Chapter 8]{horn2009},~\cite{sayed_2023} that the power matrix $A^s$ converges to $u\mathds 1^\bT$ as $s \rightarrow \infty$ at an exponential rate, where $\mathds 1$ is the vector with all entries equal to $1$ and $u$ refers to the Perron eigenvector of $A$. This is the eigenvector that is associated with the eigenvalue at $1$ and its entries are normalized to be positive and to add up to $1$:
    \begin{align}
        A u = u,\qquad u_\ell > 0,\qquad \sum_{\ell \in \mathcal N} u_\ell = 1.
    \end{align}

    At each time instant $i$, agent $k$ receives some random observation $\bzeta_{k,i}$ from the environment. In traditional social learning algorithms, it is assumed that there exists a single global true state of nature, denoted by $\theta^\star \in \Theta$. In other words, the observations at every agent $k$ are generated according to some likelihood model conditioned on this $\theta^\star$, denoted by $L_k(\zeta_{k,i} | \theta^{\star})$. The aim of the social learning algorithm then becomes to enable all agents to discover the value of $\theta^{\star}$ from the streaming observations. In this work, however, we are going to allow the truth $\theta^{\star}$ to be agent-dependent and denote it by $\theta^{\star}_k$. That is, each agent can now have its own generative model. The observations are then conditioned on the local model and we write $L_k(\zeta_{k,i} | \theta_k^\star)$. Sometimes, for compactness, we will denote the likelihood functions by writing $\boldsymbol{\zeta}_{k,i}\sim L_k(\theta_k^{\star})$ and drop the observation variables. 
    We assume the observations are independent and identically distributed (i.i.d) over time.

    The traditional social learning strategy is described as follows. Each agent $k\in\mathcal N$ assigns an initial belief $\bmu_{k,0}(\theta) \in [0, 1]$ for each state $\theta\in\Theta$ such that the total confidence sums up to $1$, i.e., $\sum_\theta \bmu_{k,0}(\theta) = 1$. In order not to exclude any hypothesis beforehand, we assume that each component of the belief vector $\bmu_{k,0}$ is strictly positive. Then, at each iteration $i$, each agent $k$ receives an observation $\bzeta_{k,i}$ and performs a local Bayesian update \cite{nedic2017fast, lalitha2018social, zhao2012learning}:
    \begin{align}
        \bpsi_{k,i}(\theta) = \frac{L_k(\bzeta_{k,i} | \theta)\bmu_{k,i-1}(\theta)}{\sum_{\theta'\in\Theta}L_k(\bzeta_{k,i}|\theta')\bmu_{k,i-1}(\theta')},\quad \forall k\in\mathcal{N}.\label{eq:adapt}
    \end{align}
    The vector $\bpsi_{k,i}$ is a probability mass function and we refer to it as the {\it intermediate} (or {\it public}) belief. The qualification ``public'' refers to the fact that this vector is shared among neighbors. For the adaptive social learning (ASL) strategy from~\cite{bordignon2020adaptive}, the update step (\ref{eq:adapt}) is replaced by\footnote{Our analysis and results are also applicable to a more general version of the ASL strategy, where an additional parameter $\beta > 0$ is introduced, and the adaptation step~(\ref{eq:adapt_adaptive}) is adjusted to $\bpsi_{k,i}(\theta) \propto L_k^\beta(\bzeta_{k,i}|\theta)\bmu^{1-\delta}_{k,i-1}(\theta)$. 
    While this form has similar convergence properties, it nevertheless allows us to recover the traditional social learning strategy~(\ref{eq:adapt}) by selecting $\beta=1$ and $\delta=0$.}:
    \begin{align}
        \bpsi_{k,i}(\theta) = \frac{L_k^\delta(\bzeta_{k,i}|\theta)\bmu^{1-\delta}_{k,i-1}(\theta)}{\sum_{\theta'\in\Theta}L^\delta_k(\bzeta_{k,i}|\theta')\bmu^{1-\delta}_{k,i-1}(\theta')},\quad \forall k\in\mathcal{N}.\label{eq:adapt_adaptive}
    \end{align}
    The step-size $\delta \in (0,1)$ infuses into the algorithm the ability to track drifts in the underlying models (hypotheses).   

    Following (\ref{eq:adapt}) or (\ref{eq:adapt_adaptive}), the agents perform geometric averaging of their neighbors public beliefs~\cite{nedic2017fast, lalitha2018social, bordignon2020adaptive}:
    \begin{align}
        &\bmu_{k,i}(\theta)=\frac{\prod_{\ell\in\mathcal{N}_k}\bpsi^{a_{\ell k}}_{\ell,i}(\theta)}{\sum_{\theta'\in\Theta}\prod_{\ell\in\mathcal{N}_k}\bpsi^{a_{\ell k}}_{\ell,i}(\theta')}, \quad \forall k\in\mathcal{N}. \label{eq:combine}
    \end{align}
    The resulting vector $\bmu_{k,i}$ is referred to as the {\it private} belief.

    At every iteration, each agent estimates the true state based on their private belief $\bmu_{k,i}$ (it can also use $\bpsi_{k,i}$):
    \begin{align}
        \widehat{\btheta}_{k,i} \triangleq \arg\max_{\theta\in\Theta}\bmu_{k,i}(\theta).
        \label{eq:truestate_est0}
    \end{align}
    Under the assumption of one global truth for all agents, both traditional and adaptive social learning have theoretical guarantees of learning the true state $\theta^{\star}$ as $i\rightarrow \infty$. Specifically, it can be shown that the belief on the true hypothesis $\theta^\star$ will converge almost surely to $1$ for the updates (\ref{eq:adapt}) and (\ref{eq:combine}) as the number of observations increases indefinitely \cite{lalitha2018social}\footnote{\label{f:2}This is true under a global identifiability assumption~\cite{lalitha2018social, bordignon2020adaptive, shumovskaia2022explainability}: for each hypothesis $\theta \neq \theta^\star$, there exists at least one agent $k\in \mathcal N$ that has strictly positive KL divergence $D_{\textup{KL}}\big(L_k\left(\theta^\star\right)||L_k\left(\theta\right)\big) > 0$. In other words, based on the received observations, the network is able to distinguish hypothesis $\theta^\star$ from any other hypothesis $\theta$.}:
    \begin{align}\label{eq:trad_conv}
        \lim_{i\rightarrow\infty} \bmu_{k,i}(\theta^\star) = 1 \textrm{, a.s.}
    \end{align}
    It also converges in probability for the updates (\ref{eq:adapt_adaptive})--(\ref{eq:combine}) as the adaptation parameter becomes smaller~\cite{bordignon2020adaptive}:
    \begin{align}
        \lim_{\delta \to 0,\;i\to\infty} \bmu_{k,i}(\theta^\star) = 1.
    \end{align}
    
\section{Truth Learning}\label{sec:truth_learning}
    \subsection{Traditional Social Learning}
        First, we examine how the performance of the  traditional social learning strategy (\ref{eq:adapt}) and (\ref{eq:combine}) is affected when there are multiple models/hypotheses for the data received by the agents. Normally, when all agents receive observations generated by the same model $\theta^\star$, it is known that the  network reaches consensus and all agents converge to that value of $\theta$ that minimizes the following objective function~\cite{lalitha2018social,bordignon2020adaptive,jadbabaie2012non,zhao2012learning}:
        \begin{align}\label{eq:div0}
            \min_\theta \sum_{k\in\mathcal N} u_k D_{\textrm{KL}}\big(L_k(\theta^\star) || L_k(\theta)\big)
        \end{align}
        Here, the notation $D_{\textrm{KL}}$ denotes the Kullback-Leibler divergence between two distributions:
        \begin{align}
            D_{\textrm{KL}} \big(L_k(\theta^\star) || L_k(\theta)\big) \triangleq \bE_{\bzeta \sim L_k(\bzeta | \theta^\star)} \log \frac{L_k(\bzeta| \theta^\star)}{L_k(\bzeta|\theta)}
        \end{align}
        Actually, the minimum value of~(\ref{eq:div0}) is $0$ for $\theta=\theta^\star$, and the KL divergence takes positive values for any other $\theta \neq \theta^\star$\footref{f:2}. This conclusion is valid when all agents in the network are assumed to receive observations that are generated by the same model $\theta^{\star}\in\Theta$. If, however, the agents happen to receive observations generated by a distribution that does not necessarily agree with any of the assumed likelihoods, i.e., if $\bzeta_k \sim f(\bzeta_k) \neq L_k(\bzeta_k | \theta)$ for any $\theta \in\Theta$, then it is known in this case that the agents will not satisfy (\ref{eq:trad_conv}) but rather they will converge to the optimal subset of hypotheses $\Theta^\star$ that minimizes the following objective function~\cite{nedic2017fast}:
        \begin{align}\label{eq:div1}
            \min_\theta \sum_{k\in\mathcal N} u_k D_{\textrm{KL}}\big(f_k(\cdot) || L_k(\theta)\big).
        \end{align}
        In other words, each agent's belief on any other hypothesis $\theta\notin\Theta^\star$ will converge to zero~\cite[Theorem 1]{nedic2017fast}:
        \begin{align}
            \lim_{i\to\infty} \bmu_{k,i}(\theta) = 0,\textrm{ a.s. } \forall \theta \notin \Theta^\star,\; \forall k \in \mathcal N.
        \end{align}
        In general, the hypotheses in the subset $\Theta^\star$ need not necessarily be optimal for any of the individual agents. In particular, there can exist some agent $k$ and some model $\theta^\circ \notin \Theta^\star$ such that  $D_{\textrm{KL}}\big(f_k(\cdot)|| L_k(\theta^\circ)\big) < D_{\textrm{KL}}\big(f_k(\cdot)|| L_k(\theta^\star)\big)$.

        The setting in this work is different.  We assume that each agent $k$ has its own truth $\theta_k^\star$ and we would like it to learn/discover that truth. According to (\ref{eq:div1}), in this case the agents in the network will converge to the consensus subset resulting from:
        \begin{align}\label{eq:div2}
            \min_\theta \sum_{k\in\mathcal N} u_k D_{\textrm{KL}}\big(L(\theta_k^\star) || L_k(\theta)\big)
        \end{align}
        Hence, under traditional social learning, even if the agents have different local truths, the network will reach consensus on some subset $\Theta^\star$. 
        
        Before stating the convergence result from~\cite{nedic2017fast}, we first define the \textit{network divergence} between hypotheses $\theta$ and $\theta'\in\Theta$:
        \begin{align}
        \label{eq:network_divergence}
            & K(\theta,\theta') \triangleq \sum_{k\in\mathcal N} u_k \bE_{\bzeta_k\sim L_k(\cdot | \theta_k^\star)} \log \frac{L_k(\bzeta_k|\theta)}{L_k(\bzeta_k|\theta')} \nonumber\\
            &= \sum_{k\in\mathcal N} u_k \Big(D_{\textrm{KL}} \big(L_k(\theta_k^\star) || L_k(\theta')\big) - D_{\textrm{KL}} \big(L_k(\theta_k^\star) || L_k(\theta)\big)\Big)
        \end{align}
        \noindent It is straightforward to verify that for any $\theta^\star$ in the optimal subset $\Theta^\star$ of (\ref{eq:div2}) and for any other $\theta \notin \Theta^\star$, the network divergence is strictly positive:
        \begin{align}\label{eq:divergence_property}
            K(\theta^\star,\theta) &= \sum_{k\in\mathcal N} u_k D_{\textup{KL}} \big(L_k(\theta_k^\star) || L_k(\theta)\big) \nonumber\\
            &\;\;\;\;  - \min_{\theta'} \sum_{k\in\mathcal N} u_k D_{\textup{KL}} \left(L_k(\theta_k^\star) || L_k(\theta')\right) > 0
        \end{align}
        \begin{Lem}[\bf{Convergence of traditional social learning~\cite[Thm. 1]{nedic2017fast}}]\label{lem:tsl_learning}
            Let $\Theta^{\star}\subset \Theta$ denote the solution to (\ref{eq:div2}). In steady-state, as $i\rightarrow \infty$, the belief of every agent $k$ at any hypothesis $\theta\notin\Theta^{\star}$ vanishes almost surely:
            \begin{align}\label{eq:tsl_lim}
                \lim_{i\to\infty} \bmu_{k,i}(\theta)=0 \textup{, a.s.}
            \end{align}
            at the following rate of convergence:
            \begin{align}
                \lim_{i\to\infty} \frac 1 i \log \frac{\bmu_{k,i}(\theta)}{\bmu_{k,i}(\theta^\star)} = -K(\theta^\star,\theta)
            \end{align}
            for any $\theta^\star \in \Theta^\star$.
        \end{Lem}
        \begin{Rem}
            Due to the Bayesian update~(\ref{eq:adapt}), the same conclusion holds for the public beliefs $\{\bpsi_{k,i}\}_{k\in\mathcal N}$.\qedsymb
        \end{Rem}
        
        \noindent It follows that when there are clusters in the graph with different true hypotheses, the algorithm behaves conservatively and forces the entire network to agree on the optimal subset $\Theta^\star$. A more desirable result would be for the agents in each cluster to converge to their own truth. Thus, depending on the setting, the traditional strategy (\ref{eq:adapt})--(\ref{eq:combine}) may or may not be the most appropriate choice for an algorithm. Next, we examine how the adaptive social learning strategy (\ref{eq:adapt_adaptive})--(\ref{eq:combine}) behaves under the same conditions.
        
    \subsection{Adaptive Social Learning}
        The main advantage of the adaptive strategy~(\ref{eq:adapt_adaptive})--(\ref{eq:combine}) is its capacity to adapt to changes in the models over time. The hyperparameter $\delta \in (0,1)$ plays an important role in this regard and defines the confidence attached to newly received samples: a smaller $\delta$ makes the algorithm lean more towards the previous true state. 
        \begin{Lem}[\bf{Log-belief ratios under adaptive social learning}]\label{lem:asl_learning}
            For each agent $k$ and for any pair of hypotheses $\theta\neq\theta'\in\Theta$, the log-belief ratio of private beliefs converges in distribution to the following random variable:
            \begin{align}\label{eq:mu_conv}
                \log \frac{\bmu_{k,i}(\theta)}{\bmu_{k,i}(\theta')} \xrightarrow[i\to\infty]{d}&\; \boldsymbol{\rho}_k(\theta,\theta') \triangleq \delta \sum_{\ell} \sum_{t=0}^\infty (1-\delta)^t [A^{t+1}]_{\ell k} \boldsymbol{{\nu}}_{\ell,t}
            \end{align}
            where
            \begin{align}
                \boldsymbol{{\nu}}_{\ell,t} \triangleq \log \frac{L_\ell(\bzeta_{\ell,t}|\theta)}{L_\ell(\bzeta_{\ell,t}|\theta')}
            \end{align}
            Its expected value is given by:
            \begin{align}\label{eq:mu_lim}
                \bE {\boldsymbol{\rho}_k(\theta,\theta')} =&\; \delta \sum_{\ell\in\mathcal N} \sum_{t=0}^\infty (1-\delta)^t [A^{t+1}]_{\ell k} \bE \boldsymbol{{\nu}}_{\ell,t}
            \end{align}
            where
            \begin{align}\label{eq:mu_exp}
                \bE \boldsymbol{{\nu}}_{\ell,t}= &\;D_{\textup{KL}}\big(L_\ell(\theta_\ell^\star) | |  L_\ell(\theta')\big) - D_{\textup{KL}}\big(L_\ell(\theta_\ell^\star) | |  L_\ell(\theta)\big)
            \end{align}
            Assuming finiteness of second-order moments for the log-likelihoods, i.e. $\bE \boldsymbol{\nu}_{\ell,t}^2 < \infty$ for any agent $\ell \in\mathcal N$, the variance of $\boldsymbol{\rho}_k$ becomes on the order of $\delta$:
            \begin{align}\label{eq:mu_var}
                \Var (\boldsymbol{\rho}_k(\theta,\theta')) = O(\delta).
            \end{align}
        \end{Lem}
        \begin{proof}
            This statement can be obtained by using~\cite[Theorem 1]{bordignon2020adaptive} and~\cite[Lemma 1]{bordignon2020adaptive}.
        \end{proof}
        \begin{Rem}
            Following similar arguments, the result can be extended to the log-belief ratio of public beliefs with $A^t$ replacing $A^{t+1}$ in~(\ref{eq:mu_conv}) and (\ref{eq:mu_lim}).
            \qedsymb
        \end{Rem}
        The results of the above lemma imply that the expectation $\bE {\boldsymbol{\rho}_k(\theta,\theta')}$ determines which hypothesis agent $k$ will prioritize in the steady-state {\it on average}. As opposed to the case when $\delta\to 0$ studied in~\cite{bordignon2020adaptive}, there is no almost sure convergence guarantee toward the prioritized hypothesis, and the log beliefs will fluctuate around their mean $\bE {\boldsymbol{\rho}_k(\theta,\theta')}$ with the variance on the order of $O(\delta)$. 
        
        These results also reveal that each agent $k$ can arrive at its own locally optimal solution because the expression on the right-hand side of (\ref{eq:mu_conv}) depends on $k$. This is in contrast to traditional social learning where it is seen from (\ref{eq:tsl_lim}) that the beliefs of all agents converge to the same zero value with a rate determined by the network divergence. In the ASL, the final inference for each agent depends on the local network, and thus on the observations and the true states of these agents: observe from~(\ref{eq:mu_lim}) that each agent gives higher importance to its close neighbors. In particular, the weight $(1-\delta)$ scales the immediate one-hop neighbors (namely, those agents $\ell$ for which $a_{\ell k}$ is non-zero), while the weight $(1-\delta)^2$ scales the agents from the 2-hop neighborhood, and so on. This way, as the value of $\delta$ increases, the influence of further connected nodes diminishes. 
        
        This observation suggests that under certain network conditions, such as community structured graphs, and for large enough $\delta$, each agent $k$ should be able to arrive on average to their own truth $\theta_k^\star$. This is because over these graphs there is a higher probability for each agent to be connected to the nodes that share the same underlying truth. We will verify that this is indeed the case. 
        
        We highlight that when $\delta\to 0$, the algorithm behaves similarly to traditional social learning and the network converges to the same subset of hypotheses $\Theta^\star$ that minimizes~(\ref{eq:div2}). 
        Particlularly, it can be shown \cite[Theorem 1]{bordignon2020adaptive} that:
        \begin{align}
            \log \frac{\bmu_{k,i}(\theta)}{\bmu_{k,i}(\theta')} \xrightarrow[\delta\to 0, i\to\infty]{\mathbb P} K(\theta, \theta')
        \end{align}
        Thus, in the following we focus on studying the behavior of adaptive social learning when it is distinct from traditional social learning, which means that we will not let $\delta\rightarrow 0$.

\section{Stochastic Block Model}\label{sec:sbm}
    The Stochastic Block Model (SBM) defines each community as a collection of agents that have a large probability of connection with each other, while the probability of connection between communities is small~\cite{abbe2017community, abbe2015exact, rohe2011spectral}. It is a popular framework for modeling social networks and graphs with polar opinions~\cite{decelle2011asymptotic, wang2022consensus}, which reflects the fact that there need not exist a single truth. For example, over social networks, groups of people may support different political parties, therefore they end up consuming media (i.e., observations) conditioned on different ``true states". In a similar manner, over a network of spatially distributed sensors, different sensors might experience different weather conditions (different temperature or precipitation). 
    This motivates us to study social learning under polar true hypotheses. Assuming that each cluster receives data conditioned on a different true state, we are interested in determining conditions that enable each of the communities to discover their own truth. From Lemma~\ref{lem:tsl_learning} we know that the traditional social learning strategies always converge to one single solution. Therefore, in this section, we focus on the adaptive social learning strategy due to its richer convergence behavior. We remark that the majority of works on community-structured graphs focus their analysis on the case of graphs with two communities~\cite{abbe2017community, abbe2015exact}. While extending the experimental part to more general cases with multiple communities is feasible, theoretical bounds often become intractable. In this work, we will similarly focus on the common scenario with two communities.
    
    We describe next the SBM model. We denote by $\bEdge$ the adjacency matrix of the network. Specifically, when the corresponding edge $(\ell, k)$ is present in the graph, the entry $[\bEdge]_{\ell,k} = 1$, and agent $k$ is able to receive signals from the neighboring agent $\ell$; otherwise, the corresponding entry $[\bEdge]_{\ell,k} = 0$. The entries of the adjacency matrix are assumed to be drawn from a Bernoulli distribution, $\bEdge \sim Bernoulli(P)$, conditioned on the probability matrix $P\in [0,1]^{|\mathcal N|\times|\mathcal N|}$ (i.e., each entry $\bEdge_{\ell, k} \sim Bernoulli(P_{\ell,k})$ is generated independently). We introduce the main idea by considering a model with two communities of sizes $n_0$ and $n_1$ such that $n_0 + n_1 = |\mathcal N|$. We denote the probabilities of connections inside the communities by $p_0$ and $p_1\in[0,1]$, and let $q_0,\;q_1 < \min\{p_0,p_1\}$ denote the probability of connection from cluster $0$ to cluster $1$, and from cluster $1$ to cluster $0$. In other words, under this model, the probability matrix $P$ takes the following block form:
    \begin{align}\label{eq:P}
        P \triangleq \left[\phantom{
            \begin{matrix}
                p_0 \\
                \vdots \\
                p_0  \\
                q_1 \\
                \vdots \\
                \underbrace{\begin{matrix}q_1\end{matrix}}_{1}
            \end{matrix}}
        \right.\hspace{-1.5em}
        \begin{matrix}
            p_0 \;\;\; \dots \;\;\; p_0 & q_0 \;\;\; \dots \;\;\; q_0 \\
            \vdots \;\;\; \ddots \;\;\; \vdots & \vdots \;\;\; \ddots \;\;\; \vdots \\
            p_0 \;\;\; \dots \;\;\; p_0 & q_0 \;\;\; \dots \;\;\; q_0 \\
            q_1 \;\;\; \dots \;\;\; q_1 & p_1 \;\;\; \dots \;\;\; p_1 \\
            \vdots \;\;\; \ddots \;\;\; \vdots & \vdots \;\;\; \ddots \;\;\; \vdots \\
            \underbrace{\begin{matrix}q_1 & \cdots & q_1\end{matrix}}_{n_0}  & \underbrace{\begin{matrix}p_1 & \cdots & p_1\end{matrix}}_{n_1} 
        \end{matrix}
        \hspace{-1.5em}
        \left.\phantom{
            \begin{matrix}
                p_0 \\
                \vdots \\
                p_0  \\
                q_1 \\
                \vdots \\
                \underbrace{\begin{matrix}q_1\end{matrix}}_{1}
            \end{matrix}}\right]\hspace{-1em}
        \begin{tabular}{l}
        $\left.\lefteqn{\phantom{
            \begin{matrix}
                p_0 \\
                \vdots \\
                p_0
            \end{matrix}}}\right\}n_0$\\
        $\left.\lefteqn{\phantom{
            \begin{matrix} b_7\\ \ddots\\ b_7\ \end{matrix}}} \right\}n_1$\\
        $\phantom{\begin{matrix} a \end{matrix}}$
        \end{tabular}
    \end{align}
    This form of $P$ allows us to generate graphs with clearly defined communities, as illustrated in Fig.~\ref{fig:edges}. For a more general SBM model, we can allow for more communities. 
    \begin{figure}
        \centering
        \begin{subfigure}[b]{0.22\textwidth}
            \centering
            \subcaption{An SBM graph model with two communities.}
            \label{fig:edges}
            \includegraphics[width=0.98\linewidth]{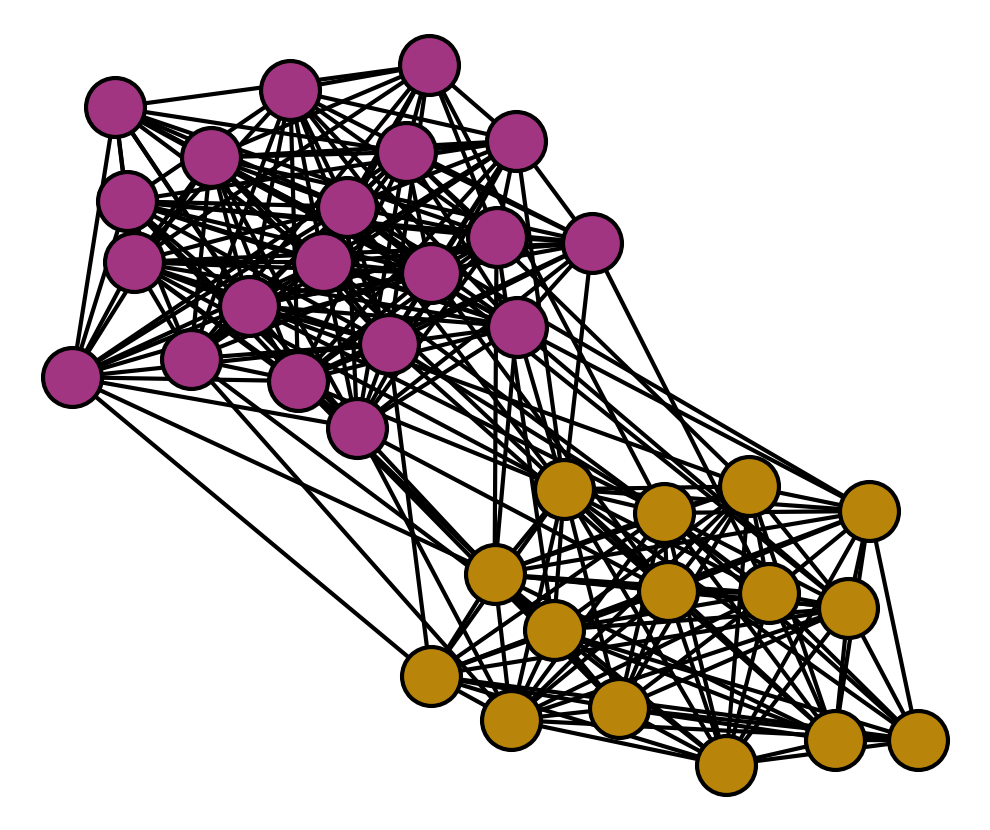}
        \end{subfigure}
        \hfill
        \begin{subfigure}[b]{0.22\textwidth}
            \centering 
            \subcaption{Combination matrix for the SBM graph in part (a).}
            \label{fig:combination}
            \includegraphics[width=0.98\textwidth]{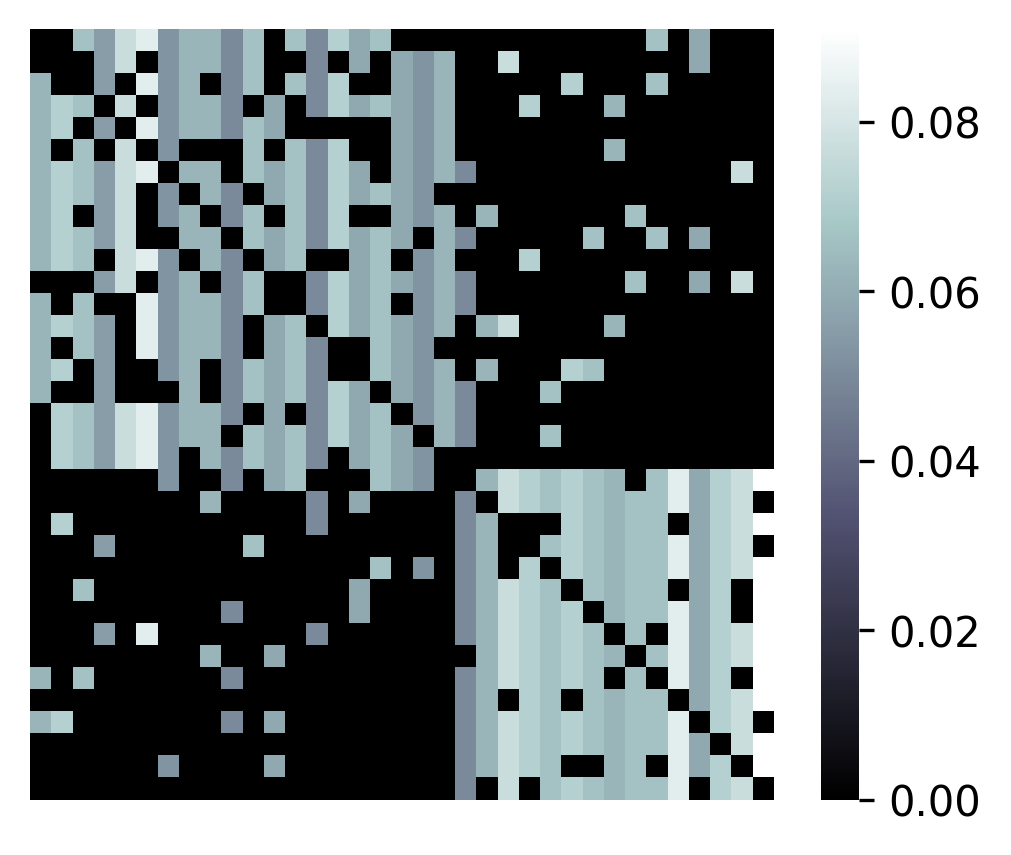}
        \end{subfigure}
        \caption{Network illustration with $n_0=20$, $n_1 = 15$, $p_0=0.8$, $p_1=0.9$, $q_0 = q_1 = 0.1$.}
        \label{fig:graph}
    \end{figure}
    
    Agents in the network will communicate with each other according to a combination protocol that is defined by some matrix $\bA$, which amounts to a weighted version of the adjacency matrix.  We assume the combination weights are set using the averaging rule~\cite{Sayed_2014}. More formally, we let $\bA = \bEdge \boldsymbol D^{-1}$ where $\boldsymbol D = \textrm{diag}(\sum_{\ell}\bEdge_{\ell,1},\dots, \sum_{\ell}\bEdge_{\ell,|\mathcal N|})$. This way, each column is normalized and agents give equal confidence to their neighbors with each entry equal to:
    \begin{align}\label{eq:combination_matrix}
        [\bA]_{\ell,k} = \frac{\bEdge_{\ell, k}}{\sum_\ell \bEdge_{\ell, k}}
    \end{align}
    We show one example of such a combination matrix in Fig.~\ref{fig:combination}, based on the adjacency network from Fig.~\ref{fig:edges}. 
    
    We additionally assume that all agents within the same community receive data arising from the same underlying model (or hypothesis). We are then interested in devising conditions under which the communities can discover their unknown models from the observations. In our particular two-communities case, we assume the binary hypotheses set $\Theta = \{\theta_0, \theta_1\}$. Thus, we let the first $n_0$ agents $k\in \C_0 \triangleq \{1,\dots,n_0\}$ to follow hypothesis $\theta_0$, i.e., $\bzeta_{k,i} \sim L_k(\cdot|\theta_0)$, and the remaining agents $k\in \C_1 \triangleq \{n_0+1,\dots,|\mathcal N|\}$ to follow hypothesis $\theta_1$. For simplicity, we assume that agents in each community (or cluster) have the same level of informativeness measured in terms of the KL divergence between the two models, as defined by the following statement.
    \begin{Asm}[\bf{Homogeneous likelihoods}]
        \label{asm:homog}
        Within each cluster $\C_i$, the agents have the same level of informativeness, i.e., for any $k \in \C_0$:
        \begin{align}
            d_0 \triangleq D_{\textup{KL}}\big(L_k\left(\theta_0\right)||L_k\left(\theta_1\right)\big)
        \end{align}
        and for any $k \in \C_1$:
        \begin{align}
            d_1 \triangleq D_{\textup{KL}}\big(L_k\left(\theta_1\right)||L_k\left(\theta_0\right)\big)
        \end{align}
        \qedsymb
    \end{Asm}
    \noindent While this assumption is not strictly necessary, it is introduced for the sake of clarity and for analytical tractability of the results. It holds, for example, when all likelihoods within each community are equal, i.e., agents receive samples from the same or similar sources.

    \subsection{SBM Network Properties}
    
        In this subsection, we study the properties of the combination matrix~(\ref{eq:combination_matrix}) generated according to the SBM model. First, we show that for a large number of agents, we can approximate the $t$-th moment of the combination matrix by the $t-$th power of its average value.
        \begin{Lem}[\bf{Expected combination matrix}]
            The expected combination matrix $\bE \bA$ is left-stochastic and takes the following block form:
            \begin{align}
                \bE \bA = \overline A + O\left(\min\{n_0,n_1\}^{-4/3}\right)
            \end{align}
            where
            \begin{align}
            \label{eq:barA}
                \overline A 
                \triangleq
                \left[
                \begin{array}{c | c} 
                      \begin{array}{c} 
                        \dfrac{p_0}{p_0n_0+q_1n_1} \mathds 1_{n_0} \mathds 1_{n_0}^\bT\\
                        \textrm{ }
                      \end{array} 
                      &
                      \begin{array}{c} 
                        \dfrac{q_0}{q_0n_0+p_1n_1} \mathds 1_{n_0} \mathds 1_{n_1}^\bT\\
                        \textrm{ }
                      \end{array} \\ 
                      \hline 
                      \begin{array}{c} 
                        \textrm{ }\\
                        \dfrac{q_1}{p_0n_0+q_1n_1} \mathds 1_{n_1} \mathds 1_{n_0}^\bT
                      \end{array}
                      & 
                      \begin{array}{c} 
                        \textrm{ }\\
                        \dfrac{p_1}{q_0n_0+p_1n_1} \mathds 1_{n_1} \mathds 1_{n_1}^\bT
                      \end{array} 
                \end{array} 
                \right]
            \end{align}
            and  $\mathds 1_{n}$ is a column vector of ones of size $n$.
            Moreover, for integer powers $t \leq |\mathcal N|$, we have:
            \begin{align}\label{eq:A_moments}
                \bE \bA^t = \overline A^t + O\left(\min\{n_0,n_1\}^{-4/3}\right)
            \end{align}
        \label{lem:A_moments}
        \end{Lem}
        \begin{proof}
            See Appendix \ref{apx:A_moments}.
        \end{proof}

        In the special case when the communities are of the same size with equal probabilities of connection inside the communities and across them, we can determine an exact expression for the matrix powers of $\overline A$.
        \begin{Lem}[\bf{Matrix powers}]
            When the communities are symmetric, i.e. $n \triangleq n_0 = n_1$, $p \triangleq p_0 = p_1$ and $q \triangleq q_0 = q_1$, the powers of $\overline A^t \approx \bE \bA^t$ for $t\leq 2n$, take the following form:
            \begin{align}
                \label{eq:barA_t}
                \overline A^t = \frac{1}{2n}
                    \begin{bmatrix}
                     1 + \left(\frac{p-q}{p+q}\right)^t & 1 - \left(\frac{p-q}{p+q}\right)^t \\
                     1 - \left(\frac{p-q}{p+q}\right)^t & 1 + \left(\frac{p-q}{p+q}\right)^t
                    \end{bmatrix}
                    \otimes  \left(\mathds 1_{n} \mathds 1_{n}^\bT\right)
            \end{align}
            where $\otimes$ denotes the Kronecker product.
        \label{lem:A_powers}
        \end{Lem}
        \begin{proof}
            See Appendix \ref{apx:A_powers}.
        \end{proof}
    
    \subsection{Truth Learning}
    
        First, we focus on the symmetric communities case from Lemma~\ref{lem:A_powers} where we have a closed-form expression for $\overline A^t$, which can be used in expression~(\ref{eq:mu_lim}). We remark first that when $\bA$ is random, as is the case in the SBM structure, the power matrix $A^{t+1}$ in~(\ref{eq:mu_lim}) should be replaced by the $(t+1)$-th moment of $\bA$. This is because the result (\ref{eq:mu_lim}) does not take into account the additional layer of randomness for combination weights, which are sampled independently from the observations. By adjusting  $\delta$, we can force each cluster to arrive at their own truth:
        \begin{Thm}[\bf{Log-belief ratios for symmetric communities}]
            Under the same scenario of Lemma~\ref{lem:A_powers}, the expected log-belief ratio from~(\ref{eq:mu_lim}) under adaptive social learning takes the following form in the steady state for any $k\in {\cal C}_0$:
            \begin{align}
                \bE {\boldsymbol{\rho}_k(\theta_0,\theta_1)} =&\; \frac 12 (d_0 - d_1) + \frac 12 \frac{\delta (d_0 + d_1) (p-q)}{p+q-(1-\delta)(p-q)} \nonumber\\
                &\;+ O\left(n^{-1/3}\right)
            \end{align}
            while for any $k\in {\cal C}_1$:
            \begin{align}
                \bE {\boldsymbol{\rho}_k(\theta_0,\theta_1)} =&\; \frac 12 (d_0 - d_1) - \frac 12 \frac{\delta (d_0 + d_1)(p-q)}{p+q-(1-\delta)(p-q)} \nonumber\\
                &\;+ O\left(n^{-1/3}\right)
            \end{align}
            where the expectations are taken with respect to the randomness in both the combination matrix $\bA$ and the observations $\bzeta_{k,i}$.
        \label{thm:log_beliefs}
        \end{Thm}
        \begin{proof}
            See Appendix \ref{apx:log_beliefs}.
        \end{proof}
    
        This result allows us to conclude that for large enough steps-size $\delta$ and large enough $n$, each cluster will converge in expectation to its own hypothesis if their corresponding $\bE {\boldsymbol{\rho}}_k(\theta_0,\theta_1)$ assumes positive or negative signs depending on the cluster, i.e.,
        \begin{subequations}
            \begin{align}
                \label{eq:delta_condition0}
                \mathcal C_0 \colon \frac 12 (d_0 - d_1) + \frac 12 \frac{\delta (d_0 + d_1)(p-q)}{p+q-(1-\delta)(p-q)} > 0 \\
                \label{eq:delta_condition1}
                \mathcal C_1 \colon  \frac 12 (d_0 - d_1) - \frac 12 \frac{\delta (d_0 + d_1)(p-q)}{p+q-(1-\delta)(p-q)} < 0
            \end{align}
        \end{subequations}
        Assuming $p > q$, these conditions can be satisfied if we select:
        \begin{align}
            \label{eq:delta_condition}
            \delta > \max\left\{\frac{d_1-d_0}{d_0}, \frac{d_0-d_1}{d_1}\right\} \times \frac {q}{p-q}
        \end{align}
        \begin{Rem}\label{rem:pub}Conditions (\ref{eq:delta_condition}) is also valid for public beliefs. Indeed, in view of~(\ref{eq:adapt_adaptive})--(\ref{eq:combine}), we observe that the log-belief ratios for $k \in \C_0$ satisfy:
            \begin{align}
                \bE \log \frac{\bpsi_{k,i}(\theta_0)}{\bpsi_{k,i}(\theta_1)} = \delta d_0 + (1-\delta) \bE \log \frac{\bmu_{k,i-1}(\theta_0)}{\bmu_{k,i-1}(\theta_1)}
            \end{align}
            Thus, if $\bE {\boldsymbol{\rho}}_k$ is positive, then the public belief is positive too, and therefore the same hypothesis is prioritized. \qedsymb
        \end{Rem}
        \begin{Rem}{
        From an information-theoretic perspective~\cite[Theorem 5]{abbe2017community}, under the same scenario of Lemma~\ref{lem:A_powers}, one cannot  exactly recover the communities if the following condition holds:
        \begin{align}\label{eq:inf_theo}
            \Bigg|\sqrt{\frac{(n/2)p}{\log (n/2)}} - \sqrt {\frac{(n/2)q}{\log (n/2)}}\Bigg| < \sqrt 2
        \end{align} 
        This implies that exact community detection becomes infeasible for highly sparse networks where $p,\; q = O(1/n)$. However, condition (\ref{eq:delta_condition}) for the choice of $\delta$ does not introduce additional assumptions on $p$ and $q$ (except that $p < q$). If we let $p$ and $q$ to be on the order of $O(1/n)$, the dependence on $n$ cancels out, allowing the social learning algorithm to perform well in such cases. We demonstrate this property in the experiments in Section~\ref{sec:computer}.
        }  \qedsymb
        \end{Rem}

        By examining~(\ref{eq:delta_condition0})--(\ref{eq:delta_condition1}), we notice that both conditions share the same first term, while the second terms with $\delta$ have the opposite signs. If we let $\delta\to 0$, then one of the inequalities will not hold. This observation suggests that for small $\delta$, the network will converge to one single hypothesis (or the subset of hypotheses) that solves~(\ref{eq:div2}) and will therefore behave similarly to traditional social learning. In contrast, a reasonably large $\delta$ would allow the less dominant cluster (i.e., cluster with smaller informativeness level $d_i$) to drive itself to their own truth.
        
        We should also remark that since $\delta \in (0,1)$, it is not always guaranteed that a $\delta$ will exist that satisfies~(\ref{eq:delta_condition}). For example, if the agents in one of the clusters have much lower level of informativeness (i.e., low $d_i$), then the right-hand side in~(\ref{eq:delta_condition}) can exceed 1. In such a case, the designer may consider reducing the connectivity between clusters. Moreover, it is important to keep in mind that a large $\delta$ enlarges the variance of the log-belief ratio (according to Lemma~\ref{lem:asl_learning}). In this case, a more accurate estimate for the expected value of the belief can be obtained by relying on the average of a window of beliefs, say, $\frac 1 M \sum_{t=i-M+1}^{i} \bmu_{k,i}(\theta)$ for some window size $M$. We illustrate this trade-off in the experiment section.

        Next we consider asymmetric communities, where $p_0$ and $p_1$ need not agree, as well as $q_0$ and $q_1$. Likewise, the sizes $n_0$ and $n_1$ can be different. 
        \begin{Thm}[\bf{Log-belief ratios for asymmetric communities}]\label{thm:asl_stepsize}
            If the following relations hold:
            \begin{align}
                \label{eq:diff_pos}
                p_0n_0d_0 - q_1n_1d_1 > 0, \quad p_1n_1d_1-q_0n_0d_0 > 0
            \end{align}
            Then, there exist a $\delta_0 \in (0,1)$, such that for any $\delta > \delta_0$, on average, each cluster converges to its own hypothesis, i.e., $\bE {\boldsymbol{\rho}_k(\theta_0,\theta_1)}$ is strictly positive or strictly negative depending on the cluster.
        \label{thm:log_beliefs_mean_het}
        \end{Thm}
        \begin{proof}
            See Appendix \ref{apx:log_beliefs_mean_het}.
        \end{proof}
        \noindent We can interpret expression~(\ref{eq:diff_pos}) as a condition wherein the internal connections are more influential than the external connections for both clusters. This influence can be regulated by factors such as community sizes, connection probabilities, or the average informativeness (i.e., $d_i$) of the agents.
        \begin{Rem}
            Similarly to Remark {\ref{rem:pub}}, the above conditions are also sufficient for public beliefs.\qedsymb
        \end{Rem}
        \begin{Exa}\label{ex:1}
            According to~(\ref{eq:delta_c0_upd})--(\ref{eq:delta_c1_upd}), for a network with $d_0 = 0.035$, $d_1 = 0.04$, $n_0=10$, $n_1=8$, $p_0=p_1=0.8$, and $q_0=q_1=0.2$, a step-size $\delta > 0.15$ is required.
        \end{Exa} 
        \begin{Exa}
            For the symmetric case where $n_0=n_1$, and $p_0=p_1$ as well as $q_0=q_1$, the bound~(\ref{eq:delta_condition}) is more precise than the bounds~(\ref{eq:delta_c0_upd})--(\ref{eq:delta_c1_upd}). Consider Example~\ref{ex:1} and let $n_0=n_1=10$. Then, (\ref{eq:delta_condition}) requires $\delta >0.05$, while (\ref{eq:delta_c0_upd})--(\ref{eq:delta_c1_upd}) require $\delta>0.11$.
        \end{Exa} 

\section{Computer Experiments}\label{sec:computer}
    \subsection{Twitter data}
        To illustrate the results from the previous sections, we consider a Twitter dataset. We will demonstrate that the adaptive social learning strategy is a more suitable fit for the opinion dynamics occuring there. Using the Twitter API and the Tweepy library, we downloaded tweets related to the Brexit discussion from 01 January 2020 to 01 April 2020 among UK parliament members. We also collected follower-followee relations to build a network of $535$ users, which is shown in Figure~\ref{fig:brexit_community}. 

        Twitter (now $\mathbb{X}$) is an example of a social media platform that can be used to illustrate the social learning paradigm. There, users form a directed network of follower-followee relations. 
        It is reasonable to assume that each user gives a certain trust level (unknown to the observer) to their followees (i.e., to the users they follow). These trust levels form the combination matrix $A$. Over Twitter, users often engage in discussions on topics of mutual interest. Every tweet can reflect a positive or negative sentiment about the topic. Thus, we can treat each tweet as a reflection of a public belief (or opinion) $\bpsi_{k,i}(\theta)$ albeit expressed in verbal form.
        
        First, we would like to illustrate that the Stochastic Block Model is a useful model that arises over social networks.  Applying the Louvain community detection algorithm~\cite{de2011generalized} to the collected graph (the algorithm requires knowledge of the adjacency matrix, whereas our algorithms mainly work with opinion dynamics without knowing the graph structure), we can observe three groups shown in Figure~\ref{fig:brexit_community}.

        Similarly to the work~\cite{shumovskaia2023discovering}, we process the text of the collected tweets using sentiment analysis  to extract public beliefs $\bpsi_{k,i}(\theta)$ from them. For this purpose, we use the Roberta-based language model~\cite{liu2019roberta, roberta_language_model} trained on 124M tweets for the sentiment analysis task. 
        We assume a binary hypothesis problem aimed at deciding whether Brexit is a good or bad idea. We analyze each parliament member tweet about Brexit during a given period of time and plot their average opinions in Figure~\ref{fig:brexit_beliefs}. 
        We see that people belonging to the same cluster tend to share similar opinions~\cite{chitra2020analyzing, dahlgren2021critical}. The traditional social learning falls short in capturing such opinion dynamics with polarized beliefs, while the adaptive strategy allows such discovery, as Theorem~\ref{thm:asl_stepsize} predicts. 
        
        \begin{figure}
            \centering
            \begin{subfigure}[b]{0.45\textwidth}  
                \centering 
                \subcaption{Communication graph based on a follower-followee relations. Node colors correspond to the output of the Louvain community detection algorithm~\cite{de2011generalized}.}
                \label{fig:brexit_community}
                \includegraphics[width=0.85\textwidth]{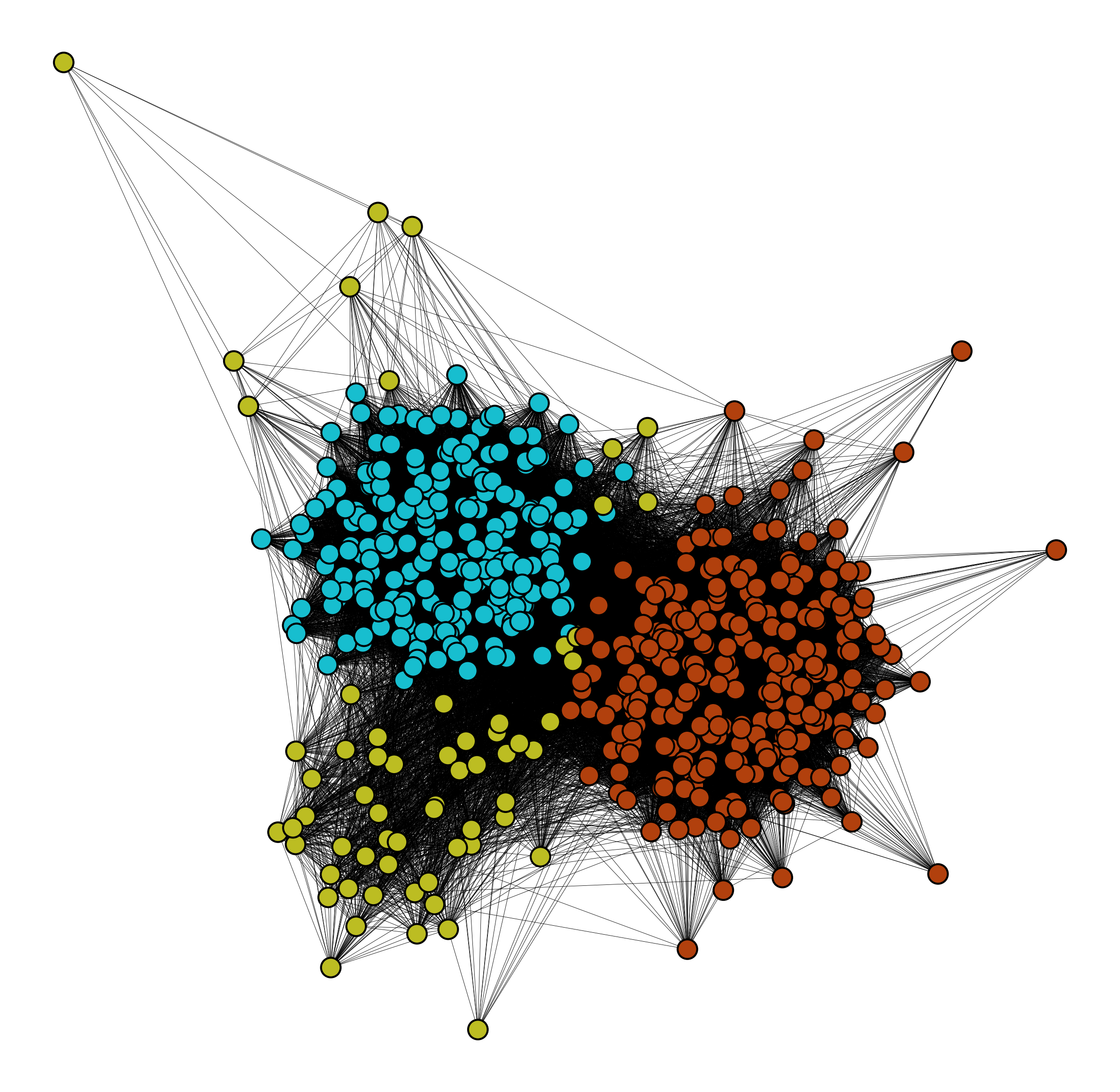}
            \end{subfigure}
            \hfill
            \begin{subfigure}[b]{0.45\textwidth}  
                \centering 
                \subcaption{Average beliefs on the hypothesis "Brexit is good`` from 01.01.2020 to 01.04.2020. The intensity of red values illustrates how positive the opinion is, and the blue reflects negativity.}
                \includegraphics[width=0.85\textwidth]{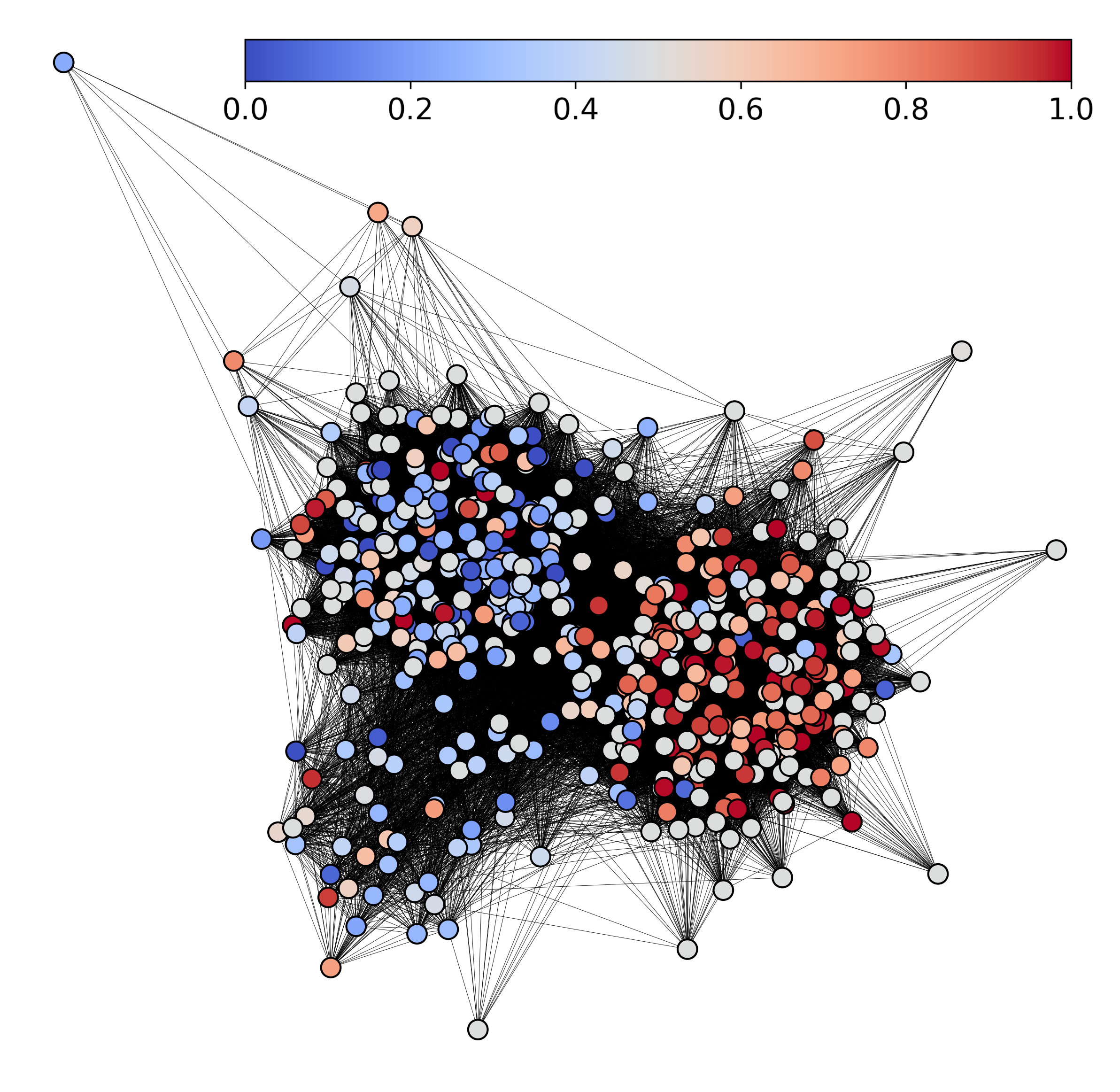}
                \label{fig:brexit_beliefs}
            \end{subfigure}
            \hfill
            \begin{subfigure}[b]{0.45\textwidth}  
                \centering
                \subcaption{The error~(\ref{eq:error_delta}) for different values of $\delta$. We examine $\delta$ from $0$ to $0.975$ with step $0.025$, and plot the error for traditional social learning in place of $\delta=0$.}
                \includegraphics[width=0.8\linewidth]{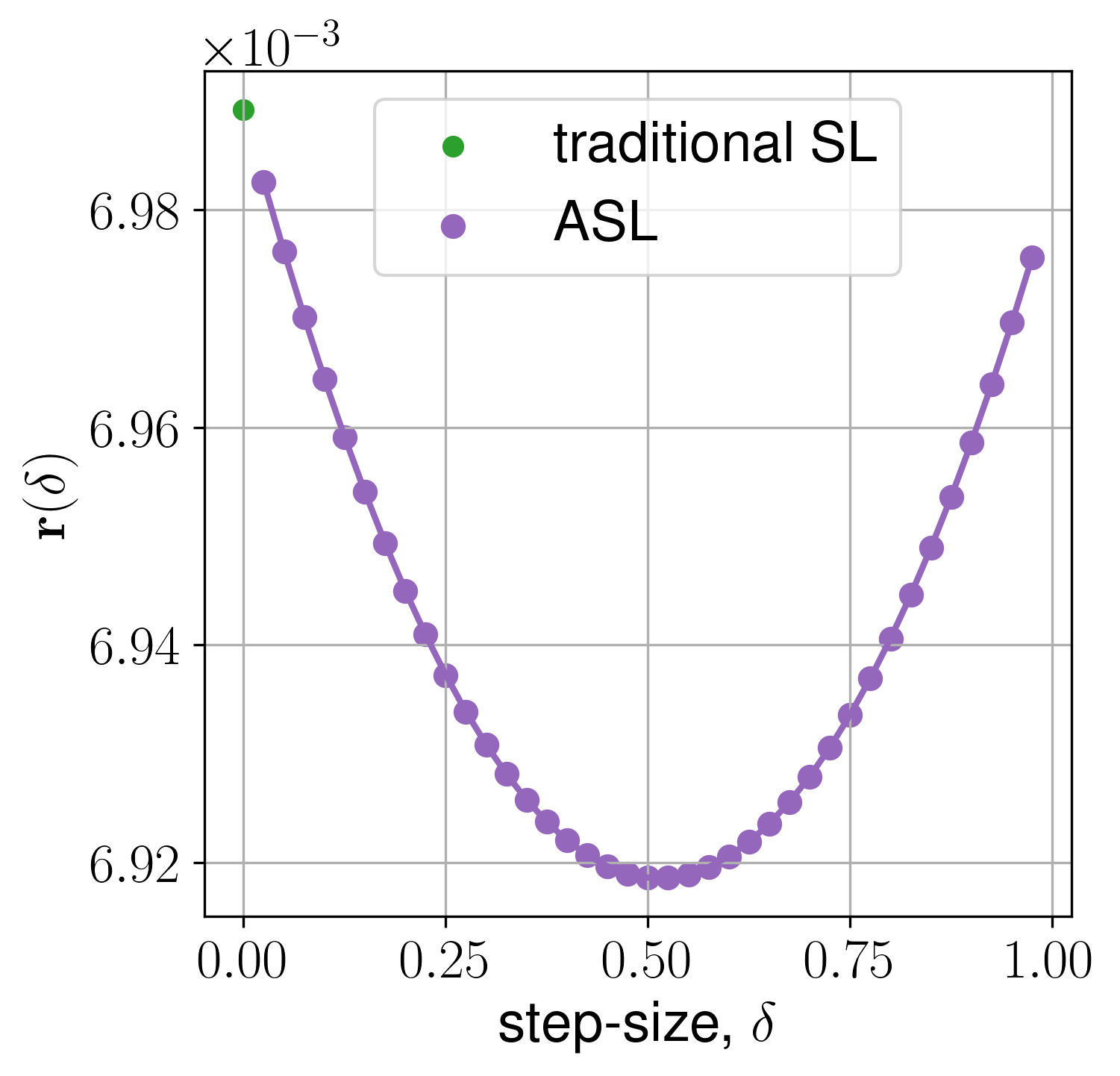}
                \label{fig:error_delta}
            \end{subfigure}
            \caption{A Twitter network of UK parliament members.}
            \label{fig:brexit}
        \end{figure}

        In the following, we will numerically show that the recursion update for the ASL~(\ref{eq:adapt_adaptive})--(\ref{eq:combine}) aligns more precisely with the collected data in contrast to the traditional strategy (\ref{eq:adapt}), (\ref{eq:combine}). First, assuming that opinions evolve according to the ASL non-Bayesian paradigm, we search for an optimal $\delta$ that fits the observed data best. Since the real data is not synchronised (i.e., people do not post their tweets at the same time), we prepare the collected data in the following manner. We assume that social learning iterations happen weekly, and if a user does not post during the chosen week, we keep the previous belief, otherwise we compute the average of the opinions extracted during that week. Thus, for each of $\mathcal N = 535$ users we have $N=13$ data points (for 13 weeks). Then, we test different $\delta \in (0,1)$ values. Fixing one $\delta$, we estimate the expected log likelihoods based on the first half of the data (first $N_1 = 6$ weeks) by using iterations (\ref{eq:adapt_adaptive})--(\ref{eq:combine}) (see~\cite[Appendix A]{shumovskaia2022explainability}). For any agent $k\in\mathcal N$, we get that: 
        \begin{align}
            \bE \log \frac{L_k(\bm\zeta_{k,i}|\theta_1)}{L_k(\bm\zeta_{k,i}|\theta_0)} \approx&\; \frac{1}{\delta(N_1 - 1)} \sum_{i=1}^{N_1 -1} 
            \Bigg( \log \frac{\bm\psi_{k,i}(\theta_1)}{\bm\psi_{k,i}(\theta_0)} \nonumber\\
            &\;- (1-\delta) \sum_{\ell} a_{\ell k} \log \frac{\bm\psi_{\ell,i-1}(\theta_1)}{\bm\psi_{\ell,i-1}(\theta_0)} \Bigg)
        \end{align}
        Here, we assume that combination weights in the adjacency matrix are set according to the averaging rule~\cite[Chapter 14]{Sayed_2014}.
        
        We use the second half of the data as validation set and compute how well the pair of estimated log-likelihoods and $\delta$ fits the recursion update (\ref{eq:adapt_adaptive})--(\ref{eq:combine}). For that purpose, we first calculate the expected log-beliefs:
        \begin{align}
            \bE \log \frac{\bm\psi_{k,i}(\theta_1)}{\bm\psi_{k,i}(\theta_0)} \approx&\; \frac 1{N-N_1} \sum_{t=N_1}^{N} 
            \log \frac{\bm\psi_{k,i}(\theta_1)}{\bm\psi_{k,i}(\theta_0)} 
        \end{align}
        
        Then, we consider the following error (again, derived from (\ref{eq:adapt_adaptive})--(\ref{eq:combine})):
        \begin{align}\label{eq:error_delta}
            &\boldsymbol{r}(\delta) = \frac 1{|\mathcal N|} \Bigg( \sum_{k\in\mathcal N} \Bigg(\bE \log \frac{\bm\psi_{k,i}(\theta_1)}{\bm\psi_{k,i}(\theta_0)} - (1-\delta) \nonumber\\
            &\;\times \sum_{\ell} a_{\ell k} \bE \log \frac{\bm\psi_{\ell,i-1}(\theta_1)}{\bm\psi_{\ell,i-1}(\theta_0)} - \delta \bE \log \frac{L_k(\bm\zeta_{k,i}|\theta_1)}{L_k(\bm\zeta_{k,i}|\theta_0)} \Bigg)^2 \Bigg)^{1/2}
        \end{align}
        In a similar manner, we evaluate how well the traditional social learning fits the data. 
        Figure~{\ref{fig:error_delta}} reveals that the ASL fits the considered data better than the traditional social learning strategies, and that the optimal value is $\delta = 0.525$.

        This experiment is one more argument that ASL is able to capture a more diverse opinion dynamics in comparison to traditional social learning, besides its tracking ability and the results from~\cite{shumovskaia2023discovering} where the algorithm was applied to another dataset from Twitter.

    \subsection{Simulated data}

    We next run experiments on simulated data.
    \subsubsection{Two communities}\label{sec:comp_twocomm}
        We consider the SBM graph shown in Figure~\ref{fig:graph_true_sbm} with connection probabilities $p_0 = p_1 = 0.8$ and $q_0 = q_1 = 0.1$. Each cluster has 15 agents and all agents have equal likelihood models with Bernoulli distributions: $L_k(\bzeta | \theta_0) = Bernoulli(0.1)$ and $L_k(\bzeta | \theta_1) = Bernoulli(0.5)$. Agents of the first block follow hypothesis $\theta_0$ and agents from the second block follow $\theta_1$. The Kullback-Leibler divergence for the first group is smaller since $0.37 = D_{\textrm{KL}}\big(L_k(\theta_0) || L_k(\theta_1)\big) < D_{\textrm{KL}}\big(L_\ell(\theta_1) || L_\ell(\theta_0)\big)=0.51$ for $k \in \C_0$ and $\ell \in \C_1$. Therefore, for small $\delta \to 0$, the network converges to $\theta_1$. And, according to~(\ref{eq:delta_condition}), one should choose $\delta > 0.056$ for each cluster to converge to their own truth.

        We illustrate the step-size condition in Figure~\ref{fig:beliefs}. We let the network follow the ASL strategy with different step-size parameters $\delta$. We see that when $\delta=0.01 < 0.056$, both log-ratios are below zero, and are close to the network divergence~(\ref{eq:network_divergence}) as was discussed after Lemma~\ref{lem:asl_learning}. When $\delta=0.1 > 0.056$, the log-belief ratio of the first group is able (on average) to stay above the zero threshold, therefore its expectation converges to hypothesis $\theta_0$. And, for larger $\delta=0.3$, we observe further increase in the gap. However, it is evident that the gap grows together with the variance, and therefore leads to an increased probability of error $\mathbb P(\arg\max_\theta \bpsi_{k,i}(\theta) \neq \theta_k^\star)$ for the second cluster in the steady state. 

        \begin{figure*}
            \centering
            \begin{subfigure}[t]{.32 \textwidth}
                \caption{True states, colors stand for hypotheses.}
                \centering
                \includegraphics[width=.95\linewidth]{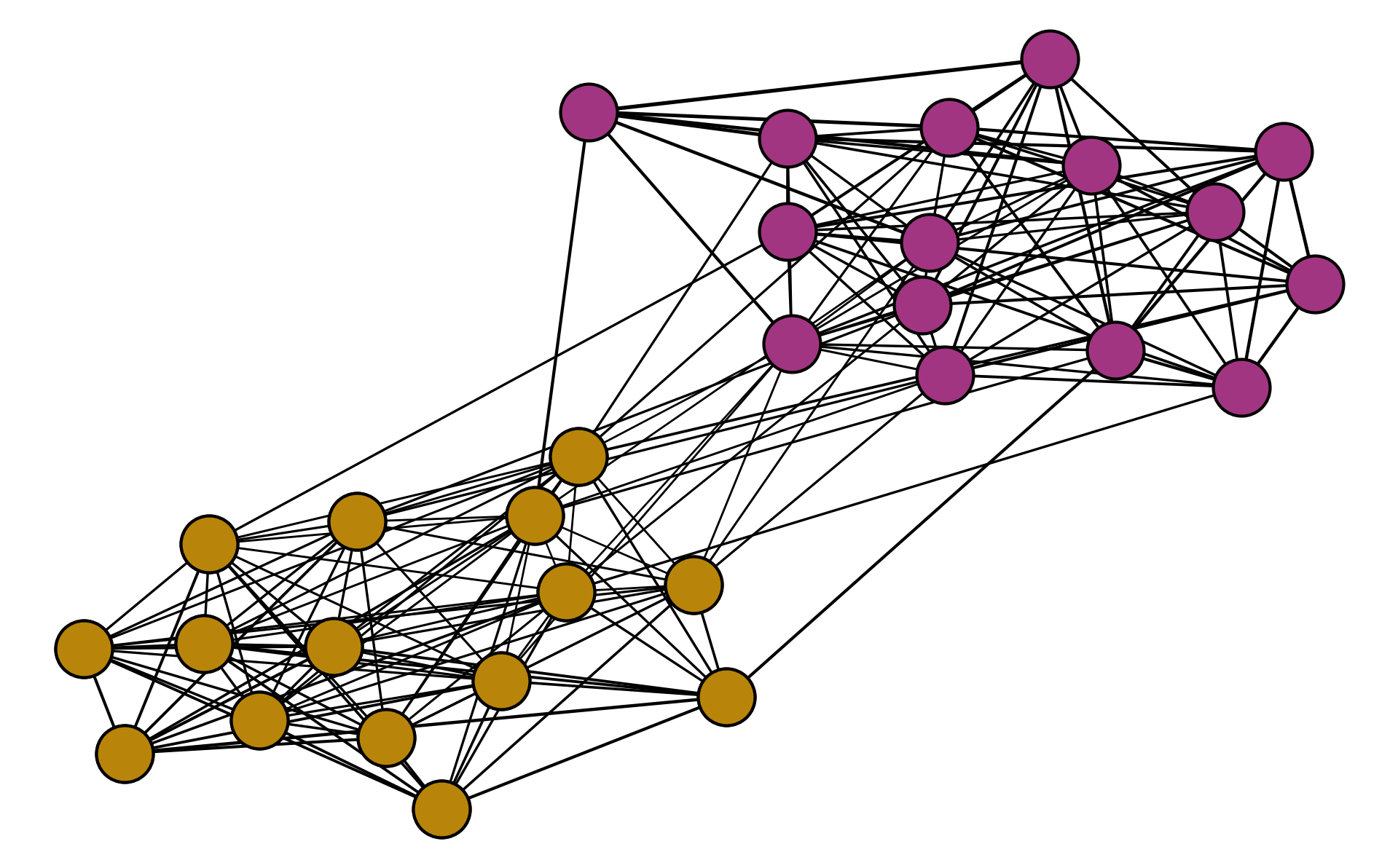}
                \label{fig:graph_true_sbm}
            \end{subfigure}
            \begin{subfigure}[t]{.32 \textwidth}
                \caption{Predicted states, $\delta = 0.01$.}
                \centering
                \includegraphics[width=.95\linewidth]{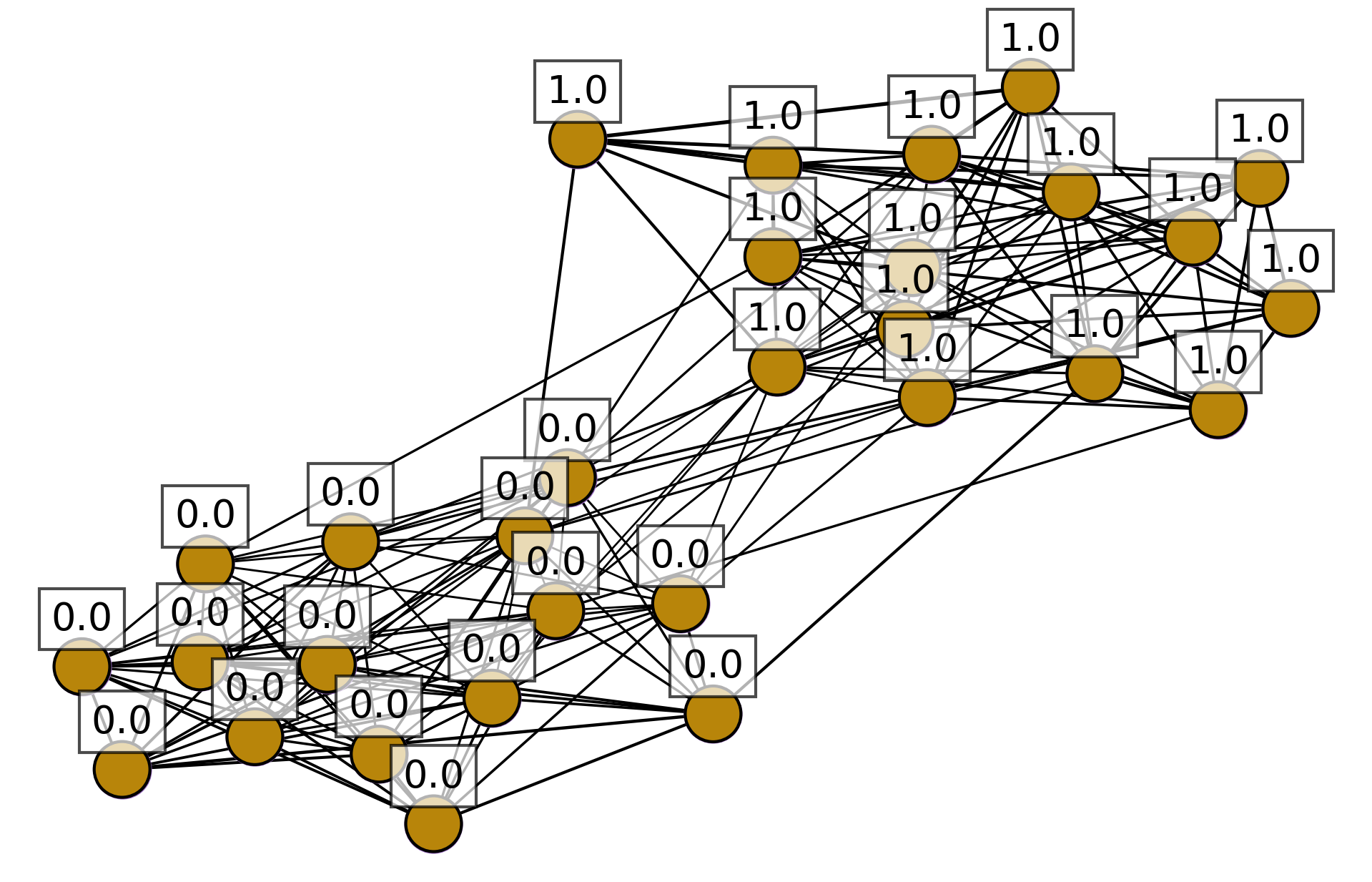}
                \label{fig:graph_est_sbm1}
            \end{subfigure}
            \begin{subfigure}[t]{.32\textwidth}
                \caption{Predicted states, $\delta = 0.1$.}
                \centering
                \includegraphics[width=.95\linewidth]{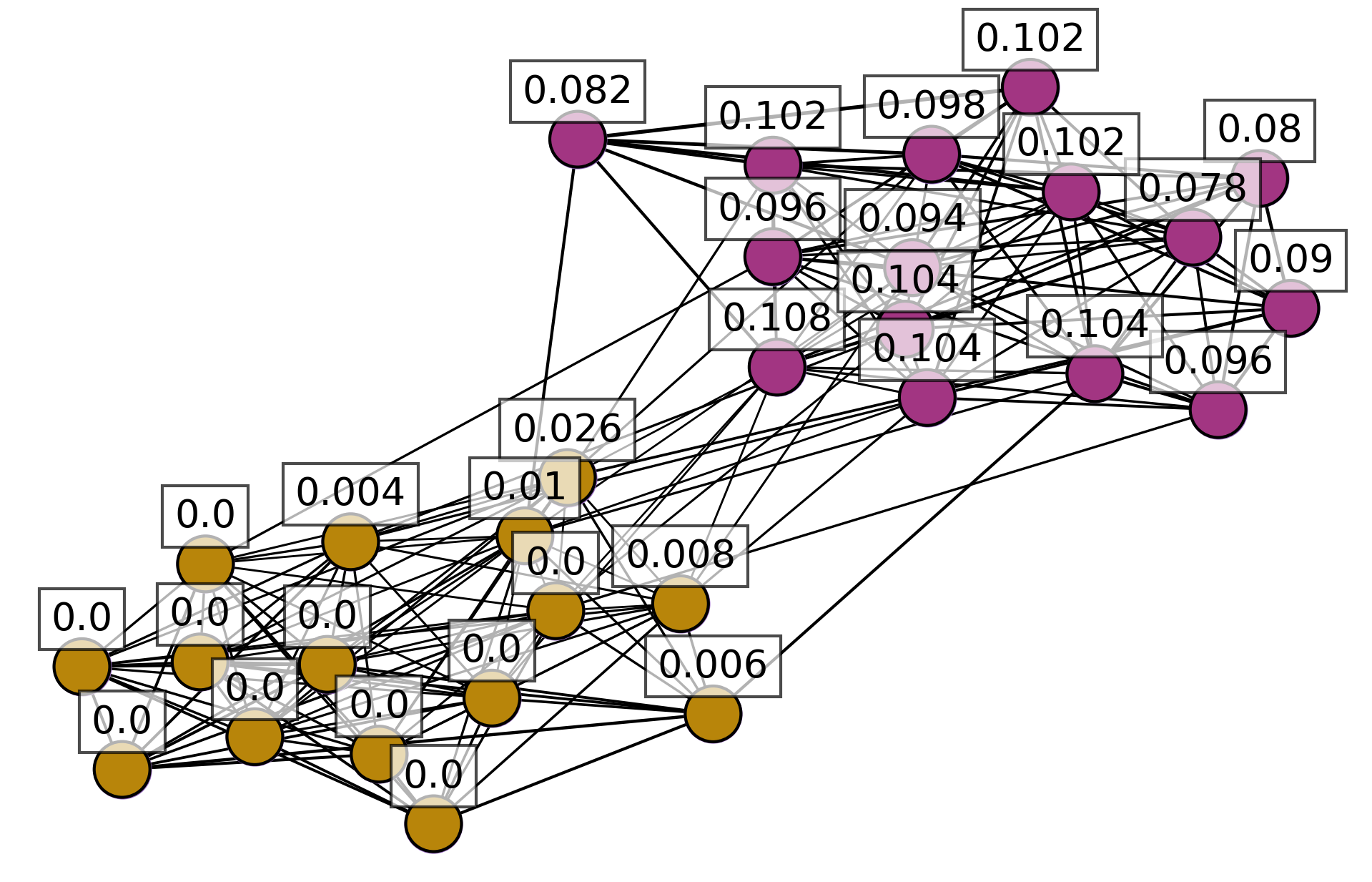}
                \label{fig:graph_est_sbm2}
            \end{subfigure}
            
            \caption{The algorithm's performance in identifying the true state of each node, using the adaptive social learning strategy. The probabilities of error (shown inside the boxes) $\mathbb P(\widehat{\boldsymbol{\theta}}_{k,i} \neq \theta_k^\star)$ are approximated based on 500 iterations.}
            \label{fig:graph_sbm}
        \end{figure*}

        \begin{figure}
            \centering
            \includegraphics[width=\linewidth]{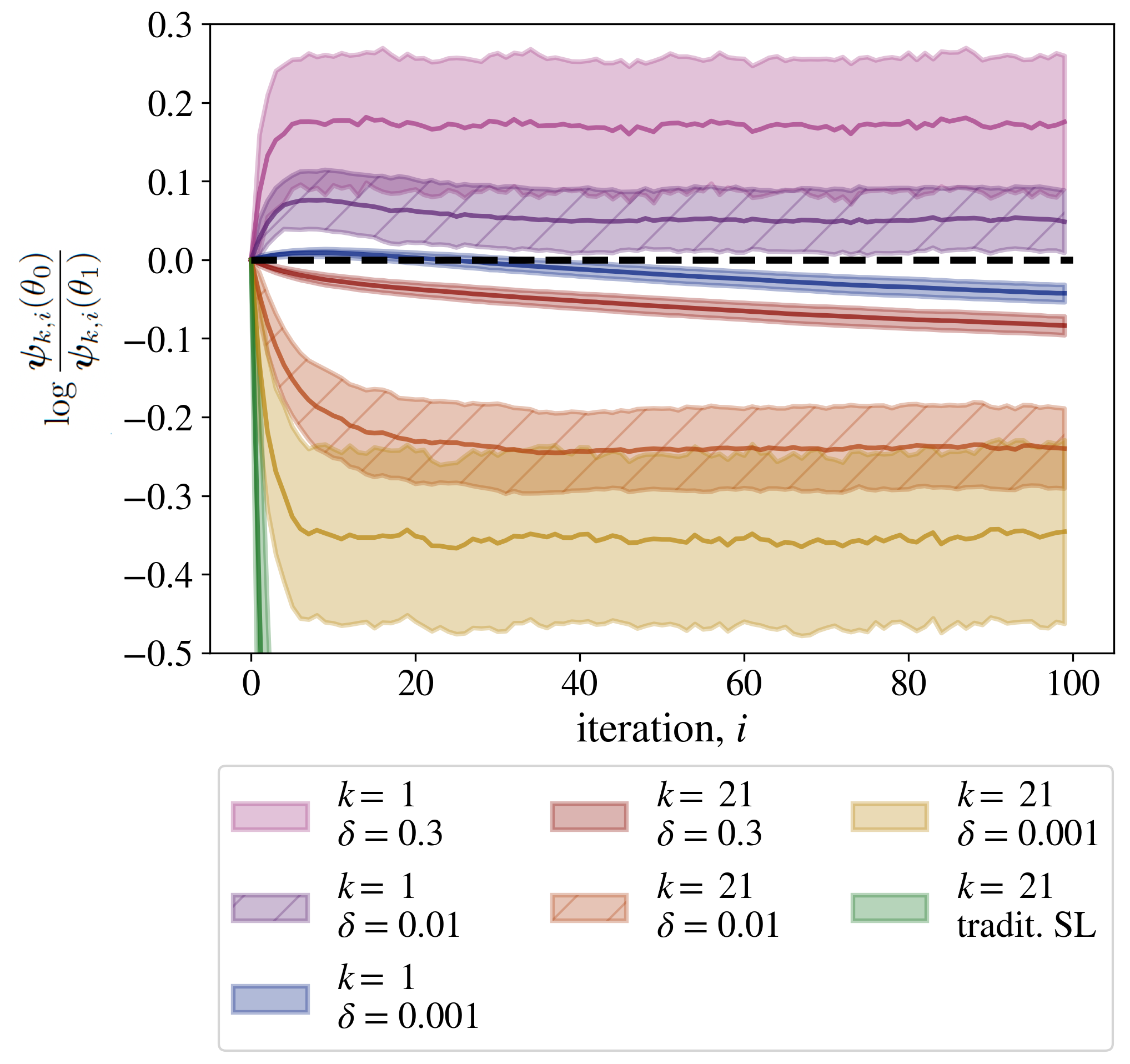}
            \caption{Evolution of log-belief ratios $\log \frac{\bpsi_{i,k}(\theta_0)}{\bpsi_{i,k}(\theta_1)}$ over time with different step-size $\delta$ (mean and standard deviations over 500 algorithm runs). Agent 1 belongs to the first cluster and follows $\theta_0$, while agent 21 belongs to the second cluster and follows $\theta_1$.}
            \label{fig:beliefs}
        \end{figure}

        Thus, we conclude that when the assumption on the cluster structure of the network holds, it makes sense to enlarge the step-size $\delta$ and learn each of the true states in a distributed manner using the ASL algorithm. 
        
    \subsubsection{Three communities of different sizes}
        We next illustrate that these conclusions hold for more general cases. We consider the SBM graph shown in Figure~\ref{fig:graph_true_sbm3} with communities of sizes $n_0 = 20,\; n_1 = 25$, and $n_2=30$. The corresponding connection probabilities are $p_0 = p_2 = 0.9$, $p_1 = 0.8$ and the between-communities probabilities are equal to $q = 0.05$. For each state, each agent's likelihood model is a multinomial distribution with 25 parameters (the probabilities were randomly generated from the uniform distribution, and then normalized accordingly). Agents of the first block follow hypothesis $\theta_0$, agents from the second block follow $\theta_1$, and agents from the third community follow hypothesis $\theta_2$. The average KL divergence (in fact, the sum of divergences between the likelihoods for $\theta_i$ and for $\theta\neq\theta_i$) per group are different: $d_0 = 1.2$, $d_1 = 1.01$, and $d_2 = 0.86$, so that $d_0 < d_1 < d_2$. Figure ~\ref{fig:graph_est_sbm3} reveals that with small $\delta = 0.01$, the network partially converges to $\theta_0$ and partially to $\theta_2$. However, with $\delta = 0.1$ (see Fig.~\ref{fig:graph_est_sbm32}), each cluster converges to their own solution with a low probability of error.

        \begin{figure*}
            \centering
            \begin{subfigure}[b]{.32 \textwidth}
                \caption{True states, colors stand for hypotheses (pink color corresponds to $\theta_0$, yellow color corresponds to $\theta_1$ and blue to $\theta_2$).}
                \centering
                \includegraphics[width=.95\linewidth]{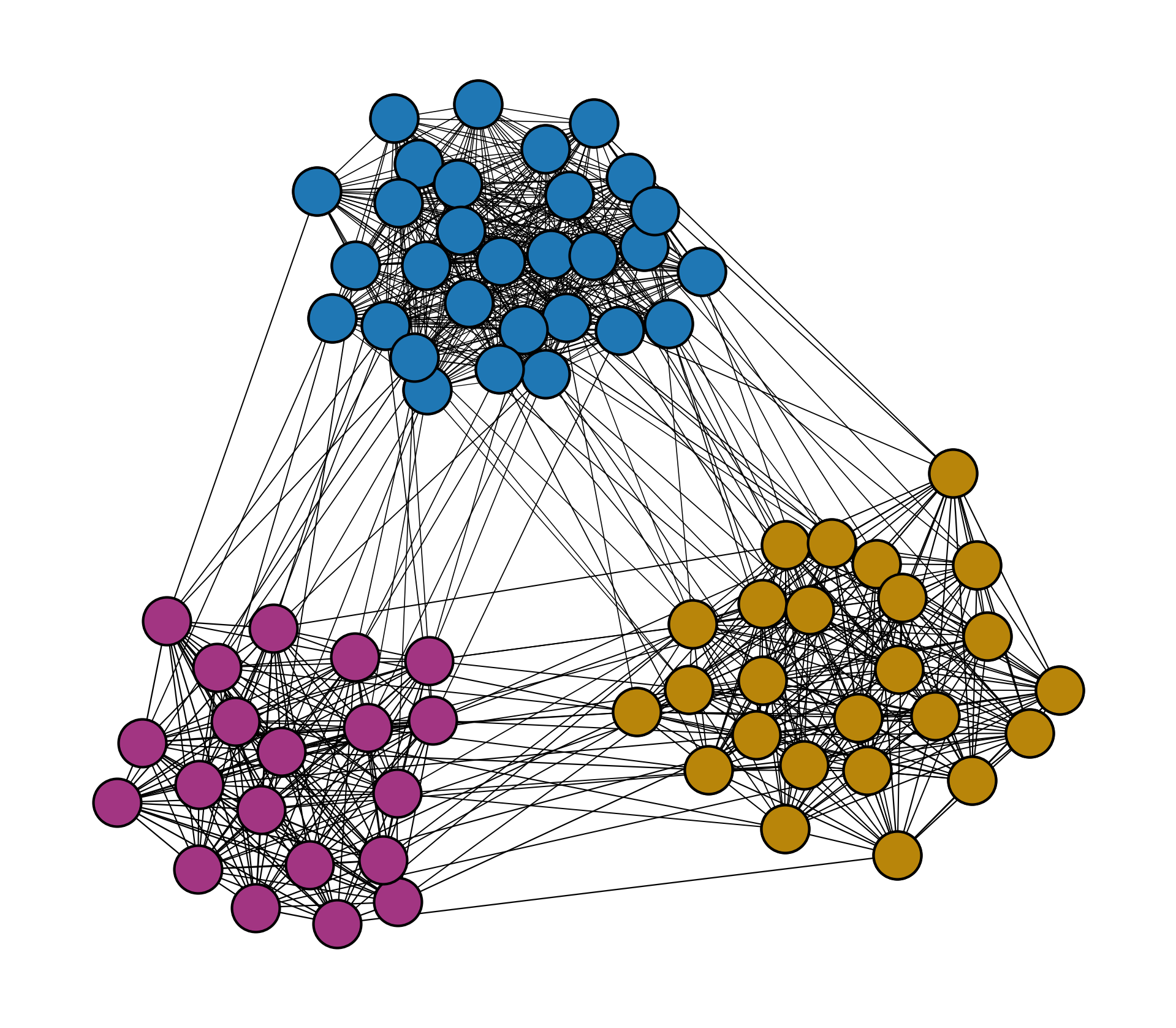}
                \label{fig:graph_true_sbm3}
            \end{subfigure}
            \begin{subfigure}[b]{.32 \textwidth}
                \caption{Predicted states, $\delta = 0.01$.}
                \centering
                \includegraphics[width=.95\linewidth]{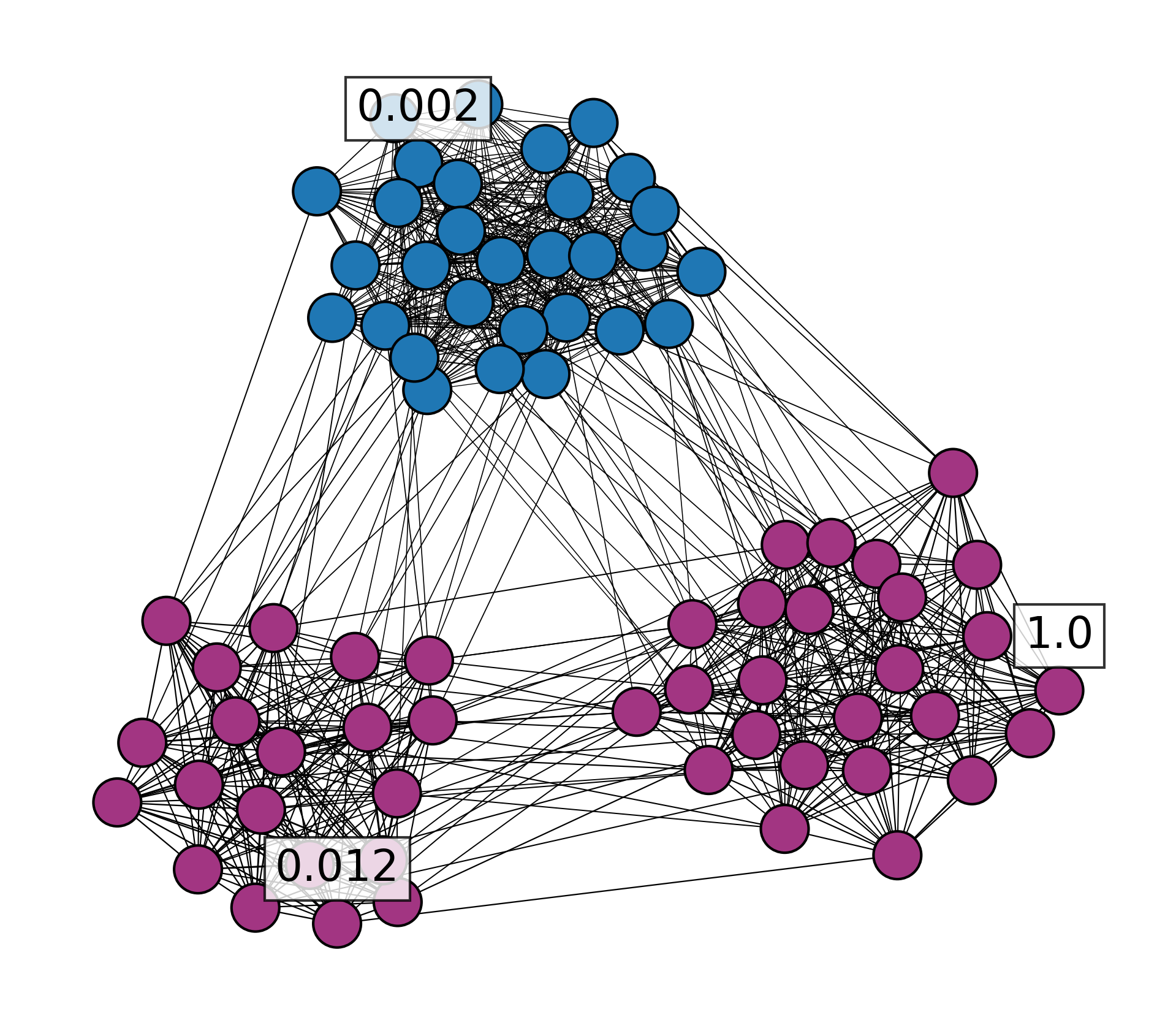}
                \label{fig:graph_est_sbm3}
            \end{subfigure}
            \begin{subfigure}[b]{.32\textwidth}
                \caption{Predicted states, $\delta = 0.1$.}
                \centering
                \includegraphics[width=.95\linewidth]{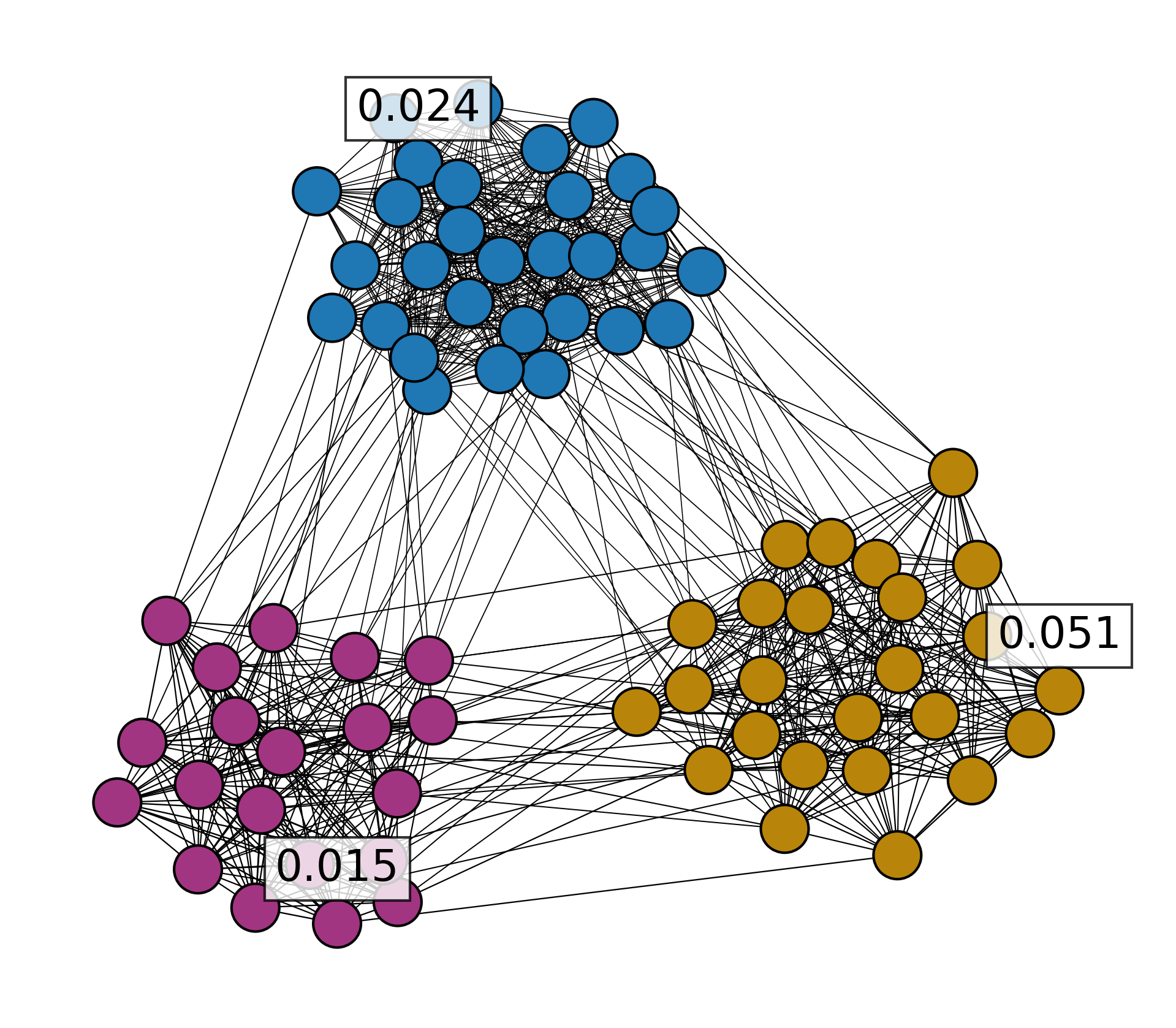}
                \label{fig:graph_est_sbm32}
            \end{subfigure}
            
            \caption{The algorithm's performance in identifying the true state of each node, using the adaptive social learning strategy. The probabilities of error $\mathbb P(\widehat{\boldsymbol{\theta}}_{k,i} \neq \theta_k^\star)$ are approximated based on 500 iterations. In boxes, we show an average error per cluster.}
            \label{fig:graph_sbm3}
        \end{figure*}

    \subsubsection{Sparse networks}
    We consider the SBM graph shown in Figure~\ref{fig:graph_true_sbm_lowp} with low connection probabilities $p_0 = p_1 = 0.25$ and $q_0 = q_1 = 0.1$. All other settings are the same as in Section~\ref{sec:comp_twocomm}. According to condition~(\ref{eq:delta_condition}), one should choose $\delta > 0.26$ for each cluster to arrive at their own truth. We can see that on average, almost all agents converge to their own hypothesis; see Fig.\ref{fig:graph_est_sbm_lowp2}. From an information-theoretic perspective, one cannot exactly recover the communities since condition~(\ref{eq:inf_theo}) is not met. Figure~\ref{fig:graph_lowp_louvain} shows the output of the Louvain community detection algorithm~\cite{de2011generalized}.

        \begin{figure*}
            \centering
            \begin{subfigure}[b]{.32 \textwidth}
                \caption{True states, colors stand for hypotheses.}
                \centering
                \includegraphics[width=.95\linewidth]{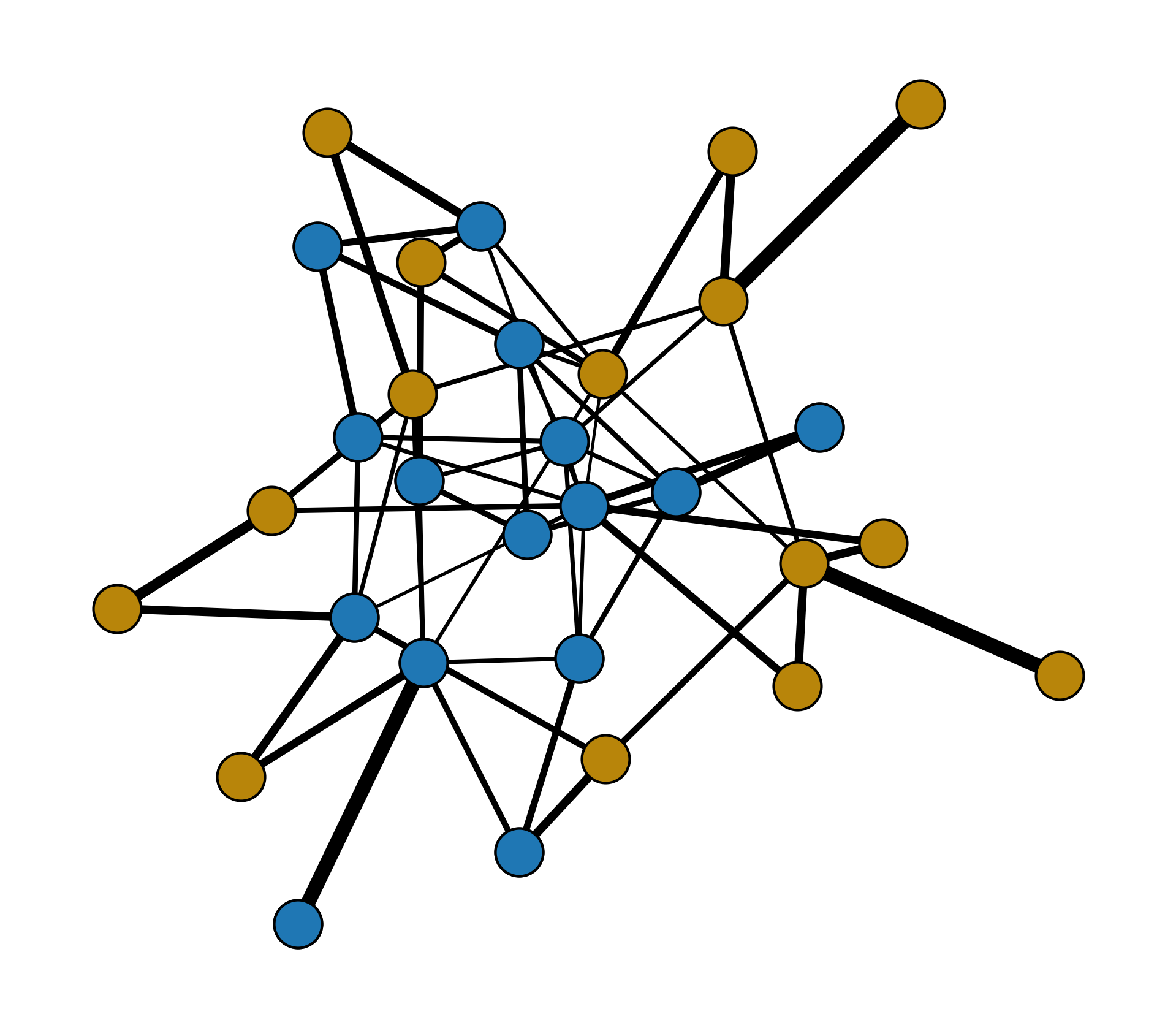}
                \label{fig:graph_true_sbm_lowp}
            \end{subfigure}
            \begin{subfigure}[b]{.32\textwidth}
                \caption{Predicted states, $\delta = 0.26$.}
                \centering
                \includegraphics[width=.95\linewidth]{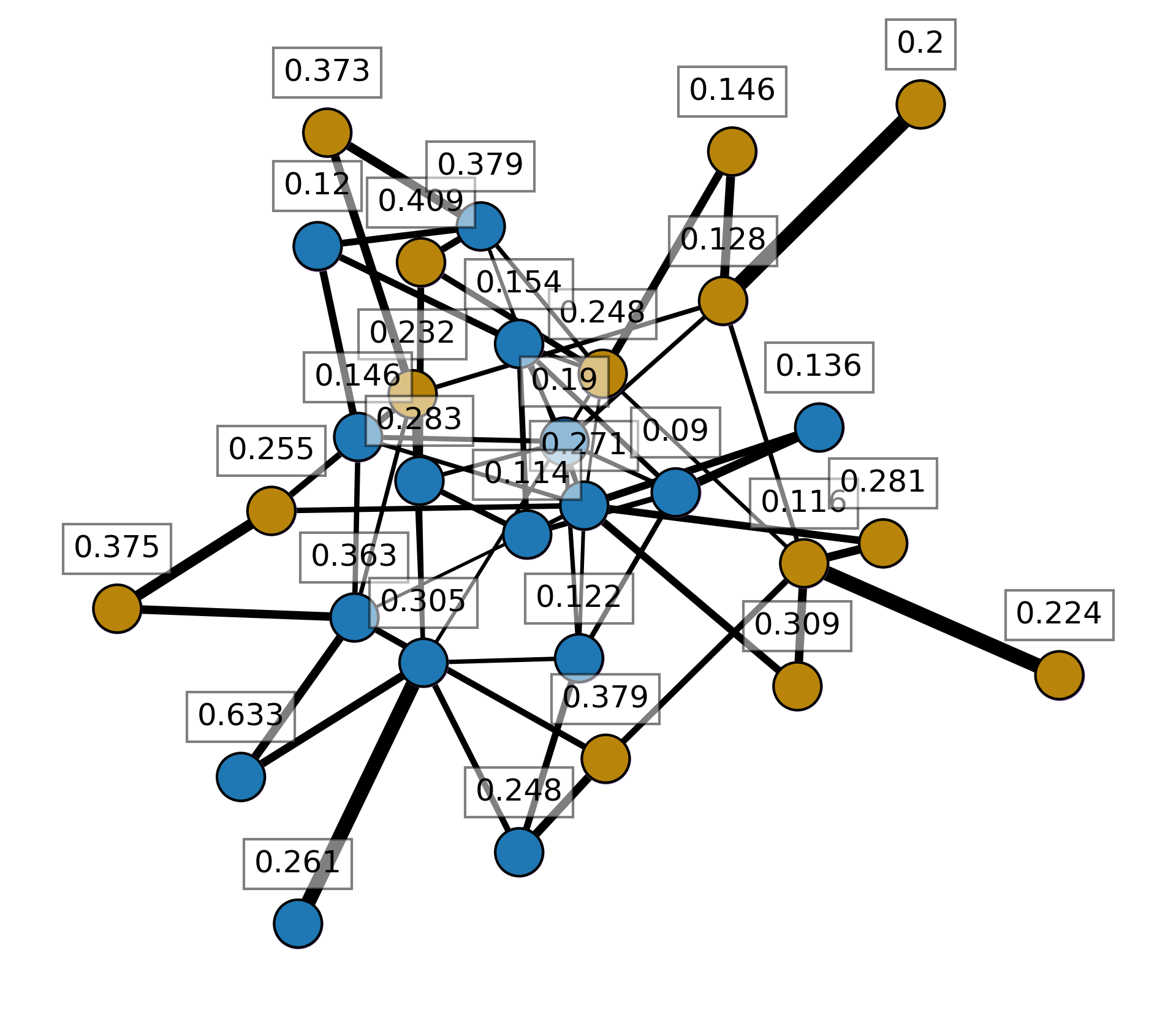}
                \label{fig:graph_est_sbm_lowp2}
            \end{subfigure}
            \begin{subfigure}[b]{.32\textwidth}
                \caption{Node colors correspond to the output of the Louvain community detection algorithm.}
                \centering
                \includegraphics[width=.95\linewidth]{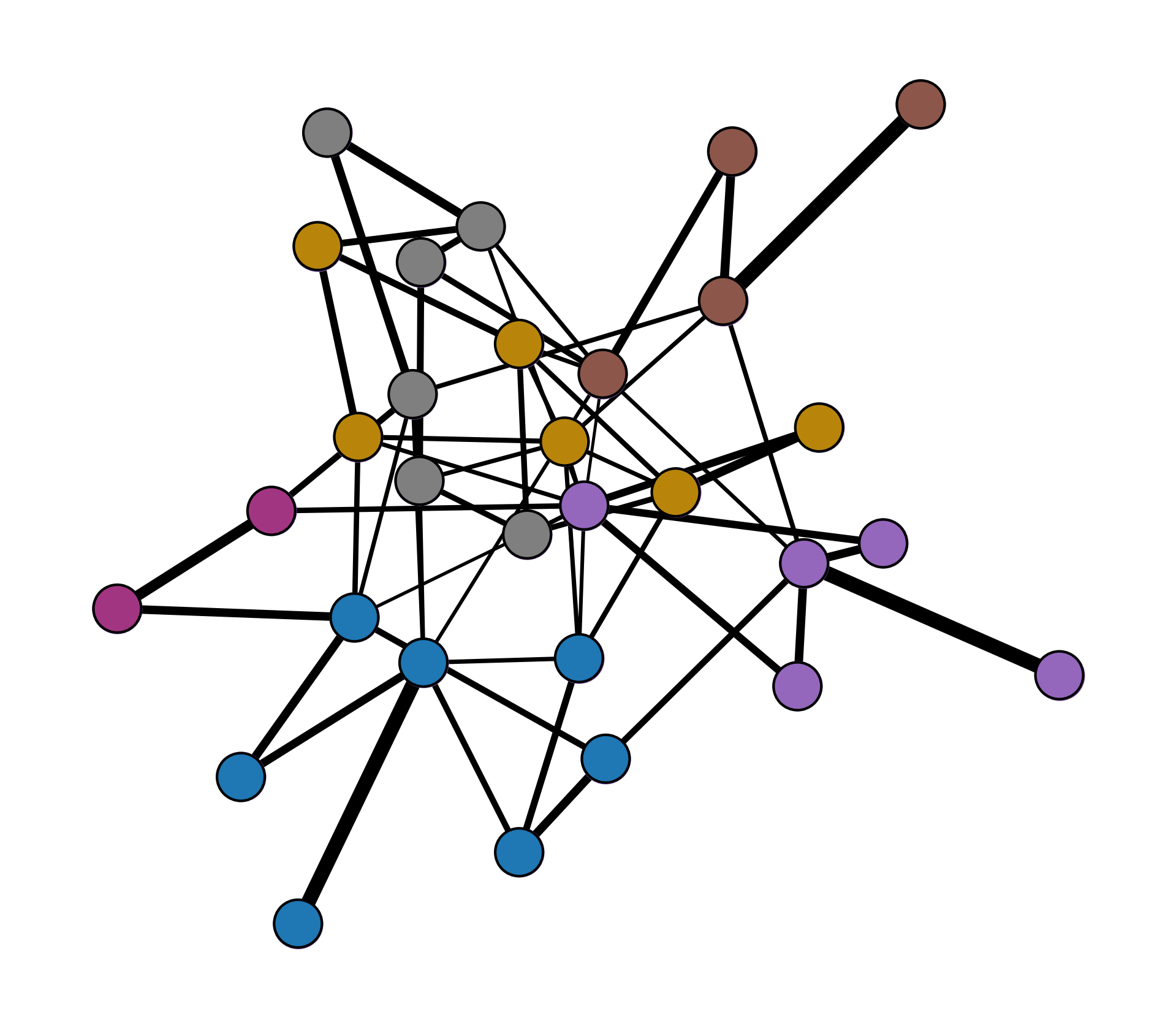}
                \label{fig:graph_lowp_louvain}
            \end{subfigure}
            
            \caption{The algorithm's performance in identifying the true state of each node, using the adaptive social learning strategy. The probabilities of error (shown inside the boxes) $\mathbb P(\widehat{\boldsymbol{\theta}}_{k,i} \neq \theta_k^\star)$ are approximated based on 500 iterations.}
            \label{fig:graph_sbm_lowp}
        \end{figure*}

\section{Conclusions}
    This work examines the behavior of social learning algorithms under multiple true hypotheses. We show that traditional social learning techniques are robust to malicious agents and the entire network converges to one solution that ``explains" the observations in the ``best way possible''. In other words, if the majority of high-informative agents operate with observations following the true hypothesis, malicious or malfunctioning agents (the agents that receive observations from another state) will not be able to drive the network to a wrong conclusion. Adaptive social learning strategies behave similarly when the adaptation hyperparameter $\delta \to 0$, however, they are the preferred choice for graphs with community structure (such as SBM). We derived values for $\delta$ for the case of two communities for each cluster to find their own truth, and illustrated the behavior of the methods through real and simulated data. 
    
\appendices

\section{Proof of Lemma~\ref{lem:A_moments}}\label{apx:A_moments}
    We start with an auxiliary result on the moments of binomial random variables.
    \begin{Lem}[\bf{Inverse moments of binomial random variables}]
        Consider any random variable $\boldsymbol b \sim B(n,p)$. For any integer $t \geq 1$ and constant $c>0$, it holds that:
        \begin{align}
            \bE \frac 1 {(c + \boldsymbol b)^t} = \frac 1 {(c + np)^t} + O\left(n^{-(t+1/3)}\right).
            \label{eq:binom_property}
        \end{align}
    \end{Lem}
    \begin{proof}
        In~\cite[Theorem 1]{arratia1989tutorial} it is shown that the deviations of binomial random variables are bounded as follows:
        \begin{align}
            \mathbb P \left(\boldsymbol b \leq an \right) \leq e^{-n H(a, p)}
            \label{eq:deviation}
        \end{align}
        where $a \in (0, p)$ and $H$ denotes the KL divergence between two Bernoulli distributions (or the relative entropy between them):
        \begin{align}
            H(a, p) \triangleq a \log \frac a p + (1-a) \log \frac {1-a}{1-p}
            \label{eq:ap_divergence}
        \end{align}
        Therefore, since $\boldsymbol b$ takes nonnegative values, it holds that for any power $t \geq 1$:
        \begin{align}
            \mathbb P \left(\boldsymbol b^t \leq (an)^t \right) \leq e^{-n H(a, p)}
        \end{align}
        Now, consider the negative moments of $\boldsymbol b$ of the form $\bE (c+ \boldsymbol b)^{-t}$ and introduce the event $\omega \triangleq \left\{ \boldsymbol b \leq an \right\}$ and its complementary event $\overline \omega$. By the property of total probability and using~(\ref{eq:deviation}), we obtain the upper bound:
        \begin{align}
            &\bE \frac 1 {(c + \boldsymbol b)^t } \nonumber\\
            &= \bE \left(\frac 1 {(c + \boldsymbol b)^t } \;\middle|\; \omega \right) \mathbb P(\omega) +  \bE \left(\frac 1 {(c + \boldsymbol b)^t } \;\middle|\; \overline \omega \right) \mathbb P(\overline \omega) \nonumber\\
            &\leq \frac 1 {c^t} \times e^{-nH(a,p)} + \frac 1{(c+an)^t} \times 1.
            \label{eq:ineq1}
        \end{align}
        Assume that for some small $\varepsilon$, we choose $a = p - \varepsilon \in (0, p)$. Then, the KL divergence~(\ref{eq:ap_divergence}) becomes:
        \begin{align}
            &H(p - \varepsilon, p) \nonumber\\
            &= (p-\varepsilon) \log \left(1 - \frac \varepsilon p \right) + (1-p+\varepsilon) \log \left(1 + \frac{\varepsilon}{1-p}\right).
        \end{align}
        Using Taylor's expansion, for small $x$, namely, $\log(1+x) = x + O(x^2)$, we proceed with:
        \begin{align}
            H(p - \varepsilon, p) &= (p-\varepsilon) \left(- \frac \varepsilon p + O(\varepsilon^2)\right)\nonumber\\
            &\;\;\;\; + (1-p+\varepsilon) \left(\frac{\varepsilon}{1-p} + O(\varepsilon^2)\right) \nonumber\\
            &= -\varepsilon + \varepsilon + O(\varepsilon^2) = O(\varepsilon^2)
        \end{align}
        We thus conclude from~(\ref{eq:ineq1}) that 
        \begin{align}
            \bE \frac 1 {(c + \boldsymbol b)^t } \leq \frac 1 {c^t} e^{-nO(\varepsilon^2)} + \frac 1{(c+(p-\varepsilon)n)^t}
        \end{align}
        Using a second Taylor's expansion for small $x$, namely, $1/ (c+x)^t = 1/c^t + O(x/c^{t+1})$, we write:
        \begin{align}
            \bE \frac 1 {(c + \boldsymbol b)^t } \leq \frac 1 {c^t} e^{-nO(\varepsilon^2)} + \frac{1}{(c + np)^t} + O(\varepsilon / n^t)
        \end{align}
        Selecting $\varepsilon = n^{-1/3}$, we get
        \begin{align}
            \bE \frac 1 {(c + \boldsymbol b)^t } &\leq \frac 1 {c^t} e^{-O(n^{1/3})} + \frac{1}{(c + np)^t} + O(n^{-(t+1/3)}) \nonumber\\
            &= \frac{1}{(c + np)^t} + O(n^{-(t+1/3)})
            \label{eq:ineq_right}
        \end{align}
        Furthermore, by Jensen's inequality, we can bound the introduced random variable from below:
        \begin{align}
            \bE \frac 1 {(c + \boldsymbol b)^t } \geq \frac 1 {(c + \bE \boldsymbol b)^t} = \frac 1 {(c + np)^t}
            \label{eq:ineq_left}
        \end{align}
        Combining~(\ref{eq:ineq_right}) and~(\ref{eq:ineq_left}), we conclude that
        \begin{align}
            \frac{1}{(c + np)^t} \leq \bE \frac 1 {(c + \boldsymbol b)^t } \leq \frac{1}{(c + np)^t} + O(n^{-(t+1/3)})
        \end{align}
        Thus, we can approximate the negative moments as in~(\ref{eq:binom_property}).
        \newline 
    \end{proof}

    Now, consider the first moment for $\bA$. Using the law of total probability, we get:
    \begin{align}
        \bE \bA_{ij} &= \bE \frac{\bEdge_{ij}}{\sum_{k=1}^{|\mathcal N|} \bEdge_{kj}} 
        = P_{ij} \bE \frac 1{1 + \sum_{k=1, k \neq i}^{|\mathcal N|} \bEdge_{kj}}.
        \label{eq:A_moments:1}
    \end{align}
    If $i, j$ lie in the same cluster $\C_0$, then,~(\ref{eq:A_moments:1}) becomes:
    \begin{align}
        \bE \bA_{ij} &\;= p_0 \bE \frac 1{1 + \sum_{k=1, k \neq i}^{|\mathcal N|} \bEdge_{kj}} = p_0 \bE \frac 1 {1 + \overline {\boldsymbol b}_0 + \boldsymbol b_1}
    \end{align}
    where $\overline {\boldsymbol b}_0 \sim B(n_0-1, p_0)$ and $\boldsymbol b_1 \sim B(n_1, q_1)$ follow Binomial distributions, and are independent of each other. This holds since each edge $\boldsymbol E_{kj}$ is independently drawn from a Bernoulli distribution according to the probability matrix $P$ in~(\ref{eq:P}). Now, introducing the new variable ${\boldsymbol b}_0 \triangleq \bEdge_{ij} + \overline{\boldsymbol{b}}_0 \sim B(n_0, p_0)$, we get:
    \begin{align}
        \bE \bA_{ij} &\;= p_0 \bE \frac 1 {\boldsymbol b_0 + \boldsymbol b_1 + (1 - \bEdge_{ij})} \nonumber\\
        &\; = p_0 \bE \frac 1 {\boldsymbol b_0 + \boldsymbol b_1} + O\left( 1/{|\mathcal N|^2} \right)
        \label{eq:A_moments:2}
    \end{align}
    In the last step, we used a third Taylor's expansion for any constant $c$ and small $x$, namely, $1/(c + \frac 1x) = x + O(x^2)$. 
    Since each edge is drawn independently, the random variables $\boldsymbol{b}_0$ and $\boldsymbol{b}_1$ are independent too, and we apply the law of total probability to~(\ref{eq:A_moments:2}):
    \begin{align}
        \bE \bA_{ij} &= p_0 \bE \left(\bE \left(\frac 1 {\boldsymbol{b}_0 + \boldsymbol b_1} \Big|  \boldsymbol{b_1}\right)\right) + O\left(1/|\mathcal N|^{-2} \right) \nonumber\\
        &\overset{(\ref{eq:binom_property}) }{=} p_0 \bE \left(\frac 1 {n_0p_0 + \boldsymbol b_1} + O(n_0^{-4/3}) \right) + O\left(1/|\mathcal N|^{-2} \right) \nonumber\\
        &\overset{(\ref{eq:binom_property}) }{=} \frac {p_0}{n_0p_0+n_1q_1} + O(n_0^{-4/3}) + O(n_1^{-4/3}) \nonumber\\
        &= \frac {p_0}{n_0p_0+n_1q_1} + O(\min\{n_0,n_1\}^{-4/3})
    \end{align}
    If $i\in\C_0$ and $j\in\C_1$ are from different clusters, we can show by similar steps that $\bE \bA_{ij}$ is given by
    \begin{align}
        \bE \bA_{ij} = \frac {q_1}{n_0p_0+n_1q_1} + O(\min\{n_0,n_1\}^{-4/3})
    \end{align}
    Thus,
    \begin{align}
        \label{eq:beba}
        \bE \bA 
        &= \overline A + O(\min\{n_0,n_1\}^{-4/3})
    \end{align}
    where $\overline A$ is a left-stochastic matrix defined in~(\ref{eq:barA}).

    Similarly, consider the second moment of $\bA$. If $i \neq j \in \C_0$ are in the same cluster, we get:
    \begin{align}
        \bE [\bA^2]_{ij} &= \sum_{m=1}^{|\mathcal N|} \bE \frac{\bEdge_{im}}{\sum_{k=1}^{|\mathcal N|} \bEdge_{km}} \frac{\bEdge_{mj}}{\sum_{k=1}^{|\mathcal N|} \bEdge_{kj}} \nonumber\\
        &= \sum_{m=1}^{|\mathcal N|} \bE \frac{\bEdge_{im}}{\sum_{k=1}^{|\mathcal N|} \bEdge_{km}} \bE \frac{\bEdge_{mj}}{\sum_{k=1}^{|\mathcal N|} \bEdge_{kj}} \nonumber\\
        &\;\;\;\; + \bE \frac{\bEdge_{ij}\bEdge_{jj}}{\left(\sum_{k=1}^{|\mathcal N|} \bEdge_{kj}\right)^2}\nonumber\\
        &\;\;\;\; - \bE \frac{\bEdge_{ij}}{\sum_{k=1}^{|\mathcal N|} \bEdge_{kj}} \bE \frac{\bEdge_{jj}}{\sum_{k=1}^{|\mathcal N|} \bEdge_{kj}} \nonumber\\
        &= [(\bE \bA)^2]_{ij} + p_0^2 \bE \frac{1}{\left(2+\sum_{k\neq{i,j}}^{|\mathcal N|} \bEdge_{kj}\right)^2}\nonumber\\
        &\;\;\;\; - p_0^2 \bE \frac{1}{1 + \sum_{k\neq i}\bEdge_{kj}} \bE \frac{1}{1+\sum_{k\neq j} \bEdge_{kj}} \nonumber\\
        &\overset{(\ref{eq:binom_property})}{=} [(\bE \bA)^2]_{ij} + p_0^2 \frac{1}{\left(2+ (n_0-2)p_0+n_1q_1\right)^2}\nonumber\\
        &\;\;\;\; - p_0^2 \frac{1}{(1 +(n_0-1)_0 + n_1q_1)^2} \nonumber\\
        &= [(\bE \bA)^2]_{ij} + O(\min\{n_0,n_1\}^{-2}) \nonumber\\
        &\overset{(\ref{eq:beba})}{=} [\overline A^2]_{ij} + O(\min\{n_0,n_1\}^{-4/3})
        \label{eq:h01}
    \end{align}
    because the random variables with different indices are independent.
    If $i = j$, then there will also be dependent random variables present in the sum:
    \begin{align}
        \bE [\bA^2]_{ii} &= \sum_{m=1}^{|\mathcal N|} \bE \frac{\bEdge_{im}}{\sum_{k=1}^{|\mathcal N|} \bEdge_{km}} \frac{\bEdge_{mi}}{\sum_{k=1}^{|\mathcal N|} \bEdge_{ki}} \nonumber\\
        &= [(\bE \bA)^2]_{ij} + \bE \frac{\bEdge_{ii}}{\sum_{k=1}^{|\mathcal N|} \bEdge_{ki}} \frac{\bEdge_{ii}}{\sum_{k=1}^{|\mathcal N|} \bEdge_{ki}} \nonumber\\
        &\;\;\;\;\; - \bE \frac{\bEdge_{ii}}{\sum_{k=1}^{|\mathcal N|} \bEdge_{ki}} \bE \frac{\bEdge_{ii}}{\sum_{k=1}^{|\mathcal N|} \bEdge_{ki}} \nonumber\\
        &= [(\bE \bA)^2]_{ij} + p_0\bE \frac{1}{(1 + \sum_{k\neq i} \bEdge_{ki})^2} \nonumber\\
        &\;\;\;\;\; - p_0^2 \bE \frac{1}{(1 + \sum_{k\neq i} \bEdge_{ki})^2}  \nonumber\\
        &\overset{(\ref{eq:binom_property})}{=} [\overline A_{ij}]^2 + \frac{p_0(1-p_0)}{(1 + (n_0-1)p_0 + n_1q_1)^2} \nonumber\\
        &\;\;\;\;\; + O(\min\{n_0,n_1\}^{-4/3}) \nonumber\\
        &= [\overline A^2]_{ij} + O(\min\{n_0,n_1\}^{-4/3}) 
        \label{eq:h1}
    \end{align}
    Thus, we have shown that
    \begin{align}
        \bE \bA^2 = \overline A^2 + O(\min\{n_0,n_1\}^{-4/3}) 
    \end{align}

    Consider next an arbitrary moment of order $t \leq |\mathcal N|$:
    \begin{align}\label{eq:f79}
        \bE [\bA^t]_{ij} = &\;\sum_{m_1=1}^{|\mathcal N|} \dots \sum_{m_{t-1}=1}^{|\mathcal N|} \bE \frac{\bEdge_{im_1}}{\sum_{k=1}^{|\mathcal N|} \bEdge_{km_1}} \nonumber\\
        &\; \times \frac{\bEdge_{m_1 m_2}}{\sum_{k=1}^{|\mathcal N|} \bEdge_{km_2}} \dots \frac{\bEdge_{m_{t-1}j}}{\sum_{k=1}^{|\mathcal N|} \bEdge_{kj}} \nonumber\\
        = &\;\sum_{m_1=1}^{|\mathcal N|} \dots \sum_{m_{t-1}=1}^{|\mathcal N|} \bE \frac{\bEdge_{im_1}}{\sum_{k=1}^{|\mathcal N|} \bEdge_{km_1}} \nonumber\\
        &\; \times \bE \frac{\bEdge_{m_1 m_2}}{\sum_{k=1}^{|\mathcal N|} \bEdge_{km_2}} \dots \bE\frac{\bEdge_{m_{t-1}j}}{\sum_{k=1}^{|\mathcal N|} \bEdge_{kj}} \nonumber\\
        &\; + \textrm{residual}
    \end{align}
    where residual takes a form of a sum over indices $m_1, \dots, m_{t-1}$ too. Each summation term in the residual includes $\bEdge_{m_im_j}^s$ of power $s\geq2$ in the numerator for at least one pair of indices $m_i,\;m_j$. We aim to show that the total sum of residual terms will be on the order of $O(|\mathcal N|^{-2})$. We follow a similar argument to~(\ref{eq:h01}) and~(\ref{eq:h1}). Clearly, each element of the residual sum is on the order of $O(|\mathcal N|^{-t}|)$. Then, we can choose $O(|\mathcal N|^2)$ combinations of $m_i$ and $m_j$. With fixed $m_i$ and $m_j$, we can choose at least $t-1-(s+1) = t-s-2$ for the remaining indices. Thus, the total number of terms with at least one $\bEdge_{m_im_j}^s$ is $O(|\mathcal N|^{t-s})$. It is sufficient to take $s=2$, therefore the entire residual is on the order of $O(|\mathcal N|^{-2})$. The equation (\ref{eq:f79}) then becomes: 
    \begin{align}
        \bE [\bA^t]_{ij} = &\; [\overline A^t]_{ij} + O(\min\{n_0,n_1\}^{-4/3}) + O(|\mathcal N|^{-2}) 
    \end{align}

\section{Proof of Lemma~\ref{lem:A_powers}}\label{apx:A_powers}

    We can represent the matrix~(\ref{eq:barA}) in the form
    \begin{align}
        \label{eq:A_p}
        \overline A = \frac 1{n(p+q)}
            \begin{bmatrix}
                p & q \\
                q & p
            \end{bmatrix}
        \otimes \left(\mathds 1_{n} \mathds 1_{n}^\bT\right)
    \end{align}
    where $ \mathds 1_{n} \mathds 1_{n}^\bT$ is ${n}\times {n}$ with each entry equal to $1$.
    Using the Kronecker product property $(A \otimes B)(C\otimes D) = (AC) \otimes (BD)$, we can represent the power matrix as:
    \begin{align}
        \label{eq:A_pow_}
        \overline A^t = \frac 1{n^t(p+q)^t}
            \begin{bmatrix}
                p & q \\
                q & p
            \end{bmatrix}^t
            \otimes \left(\mathds 1_{n} \mathds 1_{n}^\bT\right)^t
    \end{align}
    Using the eigendecomposition:
    \begin{align}  
        &\frac 1{n(p+q)}
        \begin{bmatrix}
            p & q \\
            q & p
        \end{bmatrix} \nonumber\\
        &=
        \begin{bmatrix}
            \frac 1 {\sqrt {2}} & -\frac 1 {\sqrt {2}} \\
            \frac 1 {\sqrt {2}} & \frac 1 {\sqrt {2}}
        \end{bmatrix}
        \begin{bmatrix}
            \frac 1n & 0 \\
            0 & \frac 1n \frac{p-q}{p+q}
        \end{bmatrix}
        \begin{bmatrix}
            \frac 1 {\sqrt {2}} & \frac 1 {\sqrt {2}} \\
            -\frac 1 {\sqrt {2}} & \frac 1 {\sqrt {2}}
        \end{bmatrix} \nonumber\\
        &=\frac 1 {2}
        \begin{bmatrix}
            1 & -1 \\
            1 & 1
        \end{bmatrix}
        \begin{bmatrix}
            \frac 1n & 0 \\
            0 & \frac 1n \frac{p-q}{p+q}
        \end{bmatrix}
        \begin{bmatrix}
            1 & 1 \\
            -1 & 1
        \end{bmatrix}
    \end{align}
    we obtain the power matrix of the left hand side of~(\ref{eq:A_pow_}):
    \begin{align}
        &\frac 1{n^t(p+q)^t}
        \begin{bmatrix}
            p & q \\
            q & p
        \end{bmatrix}^t\nonumber\\
        &= 
        \frac 1{2}
        \begin{bmatrix}
            1 & -1 \\
            1 & 1
        \end{bmatrix}
        \begin{bmatrix}
            \frac 1{n^t} & 0 \\
            0 & \frac 1{n^t} \left(\frac{p-q}{p+q}\right)^t
        \end{bmatrix}
        \begin{bmatrix}
            1 & 1 \\
            -1 & 1
        \end{bmatrix}\nonumber\\
        &=
        \frac 1{2n^t}
        \begin{bmatrix}
            1 + \left(\frac{p-q}{p+q}\right)^t & 1 - \left(\frac{p-q}{p+q}\right)^t \\
            1 - \left(\frac{p-q}{p+q}\right)^t & 1 + \left(\frac{p-q}{p+q}\right)^t
        \end{bmatrix}
    \end{align}
    Thus, (\ref{eq:A_pow_}) becomes:
    \begin{align}
        \overline A^t &\;= \frac{1}{2n^t}
            \begin{bmatrix}
             1 + \left(\frac{p-q}{p+q}\right)^t & 1 - \left(\frac{p-q}{p+q}\right)^t \\
             1 - \left(\frac{p-q}{p+q}\right)^t & 1 + \left(\frac{p-q}{p+q}\right)^t
            \end{bmatrix}
            \otimes  \left(\mathds 1_{n} \mathds 1_{n}^\bT\right)^t \nonumber\\
        &\;= \frac{1}{2n}
            \begin{bmatrix}
             1 + \left(\frac{p-q}{p+q}\right)^t & 1 - \left(\frac{p-q}{p+q}\right)^t \\
             1 - \left(\frac{p-q}{p+q}\right)^t & 1 + \left(\frac{p-q}{p+q}\right)^t
            \end{bmatrix}
            \otimes \left(\mathds 1_{n} \mathds 1_{n}^\bT\right)
    \end{align}

\section{Proof of Theorem~\ref{thm:log_beliefs}}\label{apx:log_beliefs}
    Consider agent $k \in \C_0$ and for simplicity of notations let ${\boldsymbol{\rho}}_k = {\boldsymbol{\rho}}_k(\theta_0,\theta_1)$. Using~(\ref{eq:mu_lim}) and Lemma~\ref{lem:A_moments}, we get:
    \begin{align}
        \label{eq:Ahh3}
        \bE {\boldsymbol{\rho}}_k &= \delta \sum_{\ell \in \mathcal N} \sum_{t=0}^\infty (1-\delta)^t \bE [ \bA^{t+1}]_{\ell k}\bE \log \frac{L_\ell(\theta_0)}{L_\ell(\theta_1)}\nonumber\\
        &= \delta \sum_{\ell \in \C_0} \sum_{t=0}^\infty (1-\delta)^t \bE [ \bA^{t+1}]_{\ell k}D_{\textrm{KL}} \big(L_\ell (\theta_0) || L_\ell (\theta_1)\big) \nonumber\\
        &\;\;\;\; - \delta \sum_{\ell \in \C_1} \sum_{t=0}^\infty (1-\delta)^t \bE [ \bA^{t+1}]_{\ell k}D_{\textrm{KL}} \big(L_\ell (\theta_1) || L_\ell (\theta_0)\big) \nonumber\\
        &= \delta \sum_{\ell \in \C_0} \sum_{t=0}^{T-1} (1-\delta)^t [\overline A^{t+1}]_{\ell k}d_0 \nonumber\\
        &\;\;\;\; - \delta \sum_{\ell \in \C_1} \sum_{t=0}^{T-1} (1-\delta)^t [\overline A^{t+1}]_{\ell k}d_1 + n O\left(n^{-4/3}\right) \nonumber\\
        &\;\;\;\; + \delta \sum_{\ell \in \C_0} \sum_{t=T}^\infty (1-\delta)^t \bE [ \bA^{t+1}]_{\ell k}d_0 \nonumber\\
        &\;\;\;\; - \delta \sum_{\ell \in \C_1} \sum_{t=T}^\infty (1-\delta)^t \bE [ \bA^{t+1}]_{\ell k}d_1 
    \end{align}
    where $T \leq n$ satisfies Lemma~\ref{lem:A_moments}. We can rewrite~(\ref{eq:Ahh3}) in the following manner:
    \begin{align}
        \bE {\boldsymbol{\rho}}_k &= \delta \sum_{\ell \in \C_0} \sum_{t=0}^\infty (1-\delta)^t [\overline A^{t+1}]_{\ell k}d_0 \nonumber\\
        &\;\;\;\; - \delta \sum_{\ell \in \C_1} \sum_{t=0}^\infty (1-\delta)^t [\overline A^{t+1}]_{\ell k}d_1 + O\left(n^{-1/3}\right) \nonumber\\
        &\;\;\;\; + \delta (1-\delta)^{T} \sum_{\ell \in \C_0} \sum_{t=T}^\infty (1-\delta)^t \bE[ \bA^{T+t} - \overline A^{T+t}]_{\ell k}d_0 \nonumber\\
        &\;\;\;\; - \delta (1-\delta)^T \sum_{\ell \in \C_1} \sum_{t=T}^\infty (1-\delta)^t \bE[ \bA^{T+t} - \overline A^{T+t}]_{\ell k}d_1 
    \end{align}
    Since both powers of $\overline A$ and $\bE \bA$ are left stochastic, their difference is bounded. Thus, due to the presence of $O(n^{-1/3})$, for each $\delta \in (0,1)$, we can choose $T \leq n$ large enough such that the terms that contain $(1-\delta)^T$ can be neglected. We omit these terms for simplicity of notation. 
    Now, proceeding with the representation for $A^t$ given by~(\ref{eq:barA_t}):
    \begin{align}
        & \bE {\boldsymbol{\rho}}_k = \frac \delta 2 \sum_{t=0}^\infty (1-\delta)^t  \left(1 + \left(\frac{p-q}{p+q}\right)^{t+1}\right)d_0\nonumber\\
        &\;\;\;\; - \frac \delta 2 \sum_{t=0}^\infty (1-\delta)^t \left(1 - \left(\frac{p-q}{p+q}\right)^{t+1}\right) d_1 + \nonumber\\
        &\;\;\;\; + O\left(n^{-1/3}\right) \nonumber\\
        &= \frac 12 \delta d_0 \frac {p-q}{p+q} \sum_{t=0}^\infty (1-\delta)^t \left(\frac{p-q}{p+q}\right)^t \nonumber\\
        &\;\;\;\; + \frac 12 \delta d_1 \frac {p-q}{p+q} \sum_{t=0}^\infty (1-\delta)^t \left(\frac{p-q}{p+q}\right)^t \nonumber\\
        &\;\;\;\; + \frac 12 \delta (d_0 - d_1) \sum_{t=0}^\infty (1-\delta)^t+ O\left(n^{-1/3}\right) \nonumber\\
        &= \frac 12 \delta (d_0 + d_1) \frac{p-q}{p+q - (1-\delta)(p-q)} \nonumber\\
        &\;\;\;\; + \frac 12 (d_0 - d_1) + O\left(n^{-1/3}\right).
    \end{align}
    Similarly, for $k \in \C_1$, we get
    \begin{align}
        \bE {\boldsymbol{\rho}}_k &= -\frac 12 \delta (d_0 + d_1) \frac{p-q}{p+q - (1-\delta)(p-q)} \nonumber\\
        &\;\;\;\; + \frac 12 (d_0 - d_1) + O\left(n^{-1/3}\right)
    \end{align}

\section{Proof of Theorem~\ref{thm:log_beliefs_mean_het}}\label{apx:log_beliefs_mean_het}
    First, we find the Perron eigenvector for $\overline A$ in~(\ref{eq:barA}), which has positive entries adding up to one and satisfies:
    \begin{align}\label{eq:barA_equation}
        \overline A u = u
    \end{align}
    Recall that $\overline A$ has a block form with equal elements inside each block, $u_1=\dots=u_{n_0}$ and $u_{n_0+1}=\dots=u_{|\mathcal N|}$. Then, equation~(\ref{eq:barA_equation}) amounts to the following linear relations:
    \begin{align}\label{eq:u0n}
        \begin{cases}
            &\frac{p_0n_0}{r_0} u_1 + \frac{q_0n_1}{r_1}u_{n_0+1} = u_1\\
            &\frac{q_1n_0}{r_0} u_1 + \frac{p_1n_1}{r_1}u_{n_0+1} = u_{n_0+1}
        \end{cases}
    \end{align}
    where we denote $r_0 = p_0n_0+q_1n_1$ and $r_1 = q_0n_0+p_1n_1$. 
    It follows that
    \begin{align}
        \label{eq:u1_un}
        u_1 = \frac{q_0r_0}{q_1r_1} u_{n_0+1}
    \end{align}
    Now since the entries of the Perron vector $u$ 
    sum up to one, we get:
    \begin{align}
        \label{eq:u1_un2}
        &n_0 u_1 + n_1 u_{n_0+1} = 1
    \end{align}
    From both~(\ref{eq:u1_un}) and~(\ref{eq:u1_un2}) we conclude that
    \begin{align}
        &u_1 = \dots = u_{n_0} = \frac{q_0r_0}{q_0r_0n_0+q_1r_1n_1} \\
        &u_{n_0+1}=\dots=u_{|\mathcal N|}=\frac{q_1r_1}{q_0r_0n_0+q_1r_1n_1} 
    \end{align}
    The power matrix $\overline A^t$ converges to $u\mathds 1^\bT$ as $t \rightarrow \infty$ at an exponential rate by the Perron-Frobenius theorem ~\cite[Chapter 8]{horn2009},~\cite{sayed_2023}. It can be shown by using simple induction arguments that due to its structure, the elements of $[\overline A^t]_{ij}$ converge to their limiting values monotonically, i.e., the diagonal block elements decrease, while non-diagonal block elements increase with $t$. Therefore, we can bound each value of $[\overline A^t]_{i,j}$ at time $t$ in the following manner. If $\ell,k \in \C_i$ belong to the same cluster, then they can be bounded by their limiting and initial values:
    \begin{align}
        \label{eq:a_bound_1}
        u_k = \frac{q_ir_i}{q_0r_0n_0+q_1r_1n_1} \leq [\overline A^t]_{\ell,k} \leq \frac{p_i}{r_i} = a_{\ell,k}
    \end{align}
    On the other hand, if $\ell,k \in \C_i, \C_j$ belong to different clusters, then
    \begin{align}
        \label{eq:a_bound_2}
        a_{\ell,k} = \frac{q_i}{r_j} \leq [\overline A^t]_{\ell,k} \leq \frac{q_ir_i}{q_0r_0n_0+q_1r_1n_1} = u_k
    \end{align}
    Since~(\ref{eq:a_bound_1}) and~(\ref{eq:a_bound_2}) are independent of $t$, the following upper and lower bounds hold for both public and private belief log-ratios. Referring to~(\ref{eq:mu_lim}), we can write for an agent $k$ in $\C_0$:
    \begin{align}
        \bE {\boldsymbol{\rho}}_k&= \delta \sum_{\ell \in \C_0}\sum_{t=0}^\infty (1-\delta)^t \bE [\bA^{t+1}]_{\ell, k} d_0 \nonumber\\
        &\;\;\;\; - \delta \sum_{\ell \in \C_1}\sum_{t=0}^\infty (1-\delta)^t \bE [\bA^{t+1}]_{\ell, k} d_1 \nonumber\\
        &= \delta \sum_{\ell \in \C_0}\bE [\bA]_{\ell,k} d_0  - \delta \sum_{\ell \in \C_1}\bE [\bA]_{\ell,k} d_1 \nonumber\\
        &\;\;\;\; + \delta \sum_{\ell \in \C_0}\sum_{t=1}^\infty (1-\delta)^t \bE [\bA^{t+1}]_{\ell, k} d_0 \nonumber\\
        &\;\;\;\;- \delta \sum_{\ell \in \C_1}\sum_{t=1}^\infty (1-\delta)^t \bE [\bA^{t+1}]_{\ell, k} d_1
    \end{align}
    Due to Lemma~\ref{lem:A_moments}, we proceed with:
    \begin{align}\label{eq:c0_frac0}
        \bE {\boldsymbol{\rho}}_k &= \delta\frac{p_0 n_0}{r_0} d_0  - \delta\frac{q_1 n_1}{r_0} d_1 \nonumber\\
        &\;\;\;\; + \delta \sum_{\ell \in \C_0}\sum_{t=1}^\infty (1-\delta)^t [\overline A^{t+1}]_{\ell, k} d_0 \nonumber\\
        &\;\;\;\;- \delta \sum_{\ell \in \C_1}\sum_{t=1}^\infty (1-\delta)^t [\overline A^{t+1}]_{\ell, k} d_1 \nonumber\\
        &\;\;\;\; + O(\min\{n_0,n_1\}^{-4/3}) \nonumber\\
        &\;\;\;\; + \delta (1-\delta)^{T} \sum_{\ell \in \C_0} \sum_{t=T}^\infty (1-\delta)^t \bE[ \bA^{T+t} - \overline A^{T+t}]_{\ell k}d_0 \nonumber\\
        &\;\;\;\; - \delta (1-\delta)^T \sum_{\ell \in \C_1} \sum_{t=T}^\infty (1-\delta)^t \bE[ \bA^{T+t} - \overline A^{T+t}]_{\ell k}d_1
    \end{align}
    Since both powers of $\overline A$ and $\bE \bA$ are left stochastic, their difference is bounded. Thus, due to the presence of $O(\min\{n_0,n_1\}^{-1/3})$, for each $\delta \in (0,1)$, we can choose $T \leq \min\{n_0,n_1\}$ large enough such that the terms that contain $(1-\delta)^T$ can be neglected. We omit these terms for simplicity of notation. 
    Now, using~(\ref{eq:a_bound_1}) and~(\ref{eq:a_bound_2}), we get:
    \begin{align}\label{eq:c0_frac}
        \bE {\boldsymbol{\rho}}_k&\geq \delta\frac{p_0 n_0}{r_0} d_0  - \delta\frac{q_1 n_1}{r_0} d_1  \nonumber\\
        &\;\;\;\; + \delta (1-\delta) \sum_{t=0}^\infty (1-\delta)^t \frac{q_0n_0r_0}{q_0n_0r_0 + q_1n_1r_1} d_0 \nonumber\\
        &\;\;\;\;- \delta (1-\delta) \sum_{t=0}^\infty (1-\delta)^t \frac{q_1n_1r_1}{q_0n_0r_0 + q_1n_1r_1} d_1 \nonumber\\
        &\;\;\;\; + O(\min\{n_0,n_1\}^{-4/3}) \nonumber\\
        &=  \delta \frac{p_0 n_0 d_0 - q_1 n_1 d_1}{r_0} \nonumber\\
        &\;\;\;\; + (1-\delta)  \frac{q_0n_0r_0d_0 - q_1n_1r_1d_1}{q_0n_0r_0 + q_1n_1r_1} \nonumber\\
        &\;\;\;\; + O(\min\{n_0,n_1\}^{-4/3}) 
    \end{align}
    Similarly, for agent $k \in \C_1$:
    \begin{align}
        \bE {\boldsymbol{\rho}}_k \leq &\; -\delta \frac{p_1 n_1 d_1 - q_0 n_0 d_0}{r_1} \nonumber\\
        &\; + (1-\delta) \frac{q_0n_0r_0d_0 - q_1n_1r_1d_1}{q_0n_0r_0 + q_1n_1r_1} \nonumber\\
        &\; + O(\min\{n_0,n_1\}^{-4/3}) 
    \end{align}
    The condition for $k\in\C_0$ to converge to hypothesis $\theta_0$ and for $k\in\C_1$ to converge to hypothesis $\theta_1$ amounts to the term $\bE {\boldsymbol{\rho}}_k$ taking positive or negative signs correspondingly. Thus, for $k\in\C_0$ and for sufficiently large $\min\{n_0,n_1\}\gg 1$, it follows from~(\ref{eq:c0_frac}) that we need $\delta$ to satisfy:
    \begin{align}
        \delta \frac{p_0 n_0 d_0 - q_1 n_1 d_1}{r_0}+ (1-\delta) \frac{q_0n_0r_0d_0 - q_1n_1r_1d_1}{q_0n_0r_0 + q_1n_1r_1}\geq 0
    \end{align}
    which transforms into
    \begin{align}\label{eq:delta_c0}
        &\delta \left( \frac{p_0 n_0 d_0 - q_1 n_1 d_1}{r_0} + \frac{q_1n_1r_1d_1 - q_0n_0r_0d_0}{q_0n_0r_0+q_1n_1r_1}\right) \nonumber\\
        &\;\;\;\;\;\;\;\;\;\;\;\;\;\;\;\;\;\;\;\;\;\;\;\;\;\;\;\;\;\;\;\;\;\;\;\;\;\;\;\;\;\;\;\;\;\geq \frac{q_1n_1r_1d_1 - q_0n_0r_0d_0}{q_0n_0r_0+q_1n_1r_1}
    \end{align}
    Taking sufficiently small $q_i$, the difference $p_0 n_0 d_0 - q_1 n_1 d_1$~(\ref{eq:diff_pos}) is positive. If the RHS of~(\ref{eq:delta_c0}) is negative, then the inequality holds for any $\delta\in(0,1)$. If the RHS is positive, then~(\ref{eq:delta_c0}) is equivalent to:
    \begin{align}\label{eq:delta_c0_upd}
        \delta \geq \tfrac{r_0(q_1n_1r_1d_1 - q_0n_0r_0d_0)}{(p_0 n_0 d_0 - q_1 n_1 d_1)(q_0n_0r_0+q_1n_1r_1) + r_0(q_1n_1r_1d_1 - q_0n_0r_0d_0)}
    \end{align}
    In other words, the sign of $q_1n_1r_1d_1 - q_0n_0r_0d_0$ determines which hypothesis is prevalent in the network. Thus, if $\theta_0$ is prevalent, then $q_1n_1r_1d_1 - q_0n_0r_0d_0$ is negative, and we do not need additional assumptions on $\delta$ to converge to $\theta_0$. Otherwise, we need~(\ref{eq:delta_c0_upd}) to be satisfied.
    Similarly, if $q_1n_1r_1d_1 - q_0n_0r_0d_0 > 0$, then for the second cluster to converge to $\theta_1$, we need $\delta$ to satisfy:
    \begin{align}\label{eq:delta_c1_upd}
        \delta \geq \tfrac{r_1(q_0n_0r_0d_0 - q_1n_1r_1d_1)}{(p_1 n_1 d_1 - q_0 n_0 d_0)(q_0n_0r_0+q_1n_1r_1) + r_1(q_0n_0r_0d_0 - q_1n_1r_1d_1)}
    \end{align}

\bibliographystyle{IEEEtran}
\bibliography{references}

\end{document}